\documentclass[letterpaper, 12pt]{amsart}
	\usepackage{amsmath, amsfonts, amsthm, mathrsfs, amssymb, graphicx, subfigure}
	\usepackage[margin = 1in]{geometry}

	\usepackage{fancyhdr}

	\usepackage{color}

	\definecolor{myblue}{rgb}{0.3, 0.0, 0.85}
	\definecolor{myviolet}{rgb}{0.5, 0.0, 0.5}

	\usepackage{hyperref}

	\hypersetup{colorlinks = true, linkcolor = myblue, citecolor = myviolet}

	\theoremstyle{plain}
	\newtheorem{defn}{Definition}[section]
	\newtheorem{prop}{Proposition}[section]
	\newtheorem{lem}{Lemma}[section]
	
	\newtheorem{cor}{Corollary}[section]
	\newtheorem{rem}{Remark}[section]

	\title{Polar Coordinates and Noncommutative Phase Space}
	\author{August J. Krueger}
	\address{Amado Building, Technion Math.\ Dept., Haifa, 32000, Israel}
	\email{ajkrueger@tx.technion.ac.il}
	\date{Oct. 22nd, 2016}

\begin{document}

	\maketitle

	\tableofcontents
	
	\begin{abstract}
	The so-called Weyl transform is a linear map from a commutative algebra of functions to a noncommutative algebra of linear operators, characterized by an action on Cartesian coordinate functions of the form $(x, y) \mapsto (X, Y)$ such that $XY -YX = i\epsilon I$, i.e. the defining relation for the Heisenberg Lie algebra. Study of this transform has been expansive. We summarize many important results from the literature. The primary goal of this work is to prove the final result: the realization of the polar transformation $(\rho, e^{i\theta}) \mapsto (R, e^{i\Theta})$ in terms of explicit orthogonal function expansions, while starting from elementary principles and utilizing minimal machinery. Our results are not strictly original but their presentation here is intended to simplify introduction to these subjects in a novel manner.\\
	\ \\
	\noindent\textsc{MSC(2010)}: 00A69, 81R60, 46L65, 35P05, 33C45.\\
	\ \\
	\noindent\textsc{Keywords}: Noncommutative geometry, harmonic analysis, Weyl transform, deformation quantization, mathematical physics.
	\end{abstract}

	
	
	

	\section{Introduction}
	
	The formation of quantum mechanics, in addition to its generous contributions to the natural sciences, brought forth a new wave of previously unforeseen mathematical correspondences and new robust frameworks for the understanding thereof. Standard quantum theory introduces an algebra of operators on an infinite dimensional vector space. The simplest case thereof occurs when considering the quantization of the classical [Hamiltonian] phase space of a point particle moving in one dimension. In this case the algebra is at least formally generated by a pair of operators $X, Y$ (usually written respectively as $Q, P$) which satisfy $[X, Y] = i\epsilon I$, where $[\cdot, \cdot]$ is the commutator, $i \in \mathbb C$ the imaginary unit, $\epsilon \in\mathbb R_+\setminus\{0\}$ (usually written as $\hbar$) a fixed parameter, and $I$ the identity operator. Intuitions from the classical theory and explicit representations of the operator algebra can be brought most to light by essentially asking the following. How one can assign a mapping that replaces $f:(x, y) \mapsto f(x, y)$, a function on $\mathbb R^2$, with an operator $F = f(X, Y)$? Use of the Weyl transform, addressed extensively in this work, is a distinctly ubiquitous answer to this question.
	
	The mathematically rigorous study of the Weyl transform has been extensive. The breadth and depth thereof has been so great that in the literature one can see frequent attempts to condense, simplify, and more clearly communicate these advances to more general audiences. We follow with this tradition in this paper by attempting to elucidate a short and straightforward guide from elementary principles of the Weyl transform through notions of polar representations of transformed functions. We claim no originality for the results presented here however feel that the production of this paper is merited by a need for a clarification and simplification of the ideas that surround the notion of cylindrical harmonics on noncommutative phase space. By this we mean an attempt at answering the following questions. If one represents $(x, y) \in \mathbb R^2$ as $\rho e^{i\theta} \in \mathbb C$ and one takes $f: (x, y) \mapsto f(x, y) = f'(\rho, \theta)$, then what is an appropriate representation of $F = f(X, Y)$ that preserves properties of the polar representation $f'(\rho, \theta)$?
	
	Generally one is concerned with the algebra generated by the elements $\{X_n, Y_n\}_{n = 1}^d$, taken with the identity $I$, that satisfy $[X_n, Y_n] = i\epsilon I$ with all other commutators vanishing. This is precisely the $2d +1$-dimensional Heisenberg Lie algebra $\mathfrak h_d$. For simplicity we restrict attention to the $d = 1$ case and we note that this fails to capture none of the larger theory since the nontrivial commutators occur pairwise. In this sense the theory strictly pertains to noncommutative planes. Much of the mathematical literature focuses instead on the Heisenberg Lie group $\mathfrak H_d$ and harmonic analysis thereupon, see e.g. the excellent texts by Thangavelu \cite{Th98}\cite{Th04}. Often in this perspective one considers $\mathfrak H_d$ as isomorphic to $\mathbb R^{2d +1}$ as a set that carries an algebraic product, with $\mathfrak h_n$ a set of vector fields with a Lie bracket. This approach has been very fruitful with many mathematical applications, c.f. for example the comprehensive survey by Howe \cite{Ho80}. We believe, however, that this ``top down'' approach has been less successful at taking root in more distant branches of mathematics and especially in the natural sciences. We then take the ``bottom up'' approach whereby one considers the algebra generated directly by $\mathfrak h_d$, usually taken with a unitary representation on an infinite dimensional Hilbert space. We consider Folland's excellent text \cite{Fo89} to be the de facto standard text on this approach.
	
	The replacing of a commutative algebra with a noncommutative algebra for the modeling of quantum physics is a substantial pillar of the set of ideas known broadly as quantization. The use of the Weyl transform for this end is known as deformation quantization, which has seen much use for semi-classical approximations and quantum optics. These methods from physics have also been carried over into pure math with micro-local analysis, symbol calculus, and a general study of pseudo-differential operators. There has recently been an interest, in the physics community, in the use of the replacement $(x_n, y_n) \mapsto (X_n, Y_n)$ to modify general PDE on manifolds. This was first strongly encouraged in the context stabilizing the solutions of effective field theories for string theories by Gopakumar, Minwalla, and Strominger in \cite{GoMiSt00}, the ideas of which were brought closer to those of the mathematics community by Chen, Fr\"ohlich, and Walcher in \cite{ChFrWa03}. These works tend to rely on the Moyal star product approach, which introduces a particular highly noncommutative product on otherwise commutative functions. This method requires fewer new mathematical objects but can be rather opaque for the sake of analysis. For a more extensive discussion of these, among others, approaches and applications see, for example, \cite{KrSo16}. We hope our presentation in this paper will encourage members of the high energy physics community to more utilize the Weyl transform.
	
	In this paper we will focus on exact, explicit representations of polar forms (especially cylindrical harmonics) in terms of well-known families of functions. This emphasis is also not original. The first explicit approach to noncommutative cylindrical harmonics in the literature can be seen in Theorem 4.2 of \cite{Ge84} and calculations contained in its proof. It is worth noting that \cite{Ge84} dedicates its time up until the theorem in question reviewing the general theory of the Weyl transform in the context of the Heisenberg group and therefore interested readers may seek a very different presentation there. Connections between spherical (not cylindrical) harmonics and orthogonal polynomial expansions can be found in \cite{Gr80}. We consider the clear and direct texts of Thangavelu \cite{Th93} and \cite{Th98} as being very close to the treatments of orthogonal polynomial expansions and polar forms. We therefore consider these texts to be generous extensions of the subject matter of this paper from a different mathematical perspective, e.g. a focus on $\mathfrak H_n$ instead of $\mathfrak h_n$, but rather similar in spirit. We note that many of the Cartesian and polar function expansions we consider here can be seen in the literature of the quantum optics community, particularly through the so-called Laguerre-Gauss polynomials. One can see, for example, various properties discovered in \cite{KaEA14}, seemingly independently from the more recent mathematical literature.
	
	We are motivated to produce this work due to a perceived need for a new summary of results geared toward particular applications. Recently Acatrinei in \cite{Ac13} explored the notion of polar functions in the sense mentioned above from the perspective of the physics literature. The observations made there about noncommutative cylindrical harmonics were insightful and applications thereby to new results on local decay estimates for Schr\"odinger operators was considered by Kostenko and Teschl in \cite{KoTe1602}. We feel that the perspectives outlined in this paper would be of great benefit to further research along these lines. For simplicity we only consider $\mathfrak h_1$, e.g. the 3-dimensional Heisenberg Lie algebra. In the penultimate section we address the main results in focus, explicit representations of the map $(\rho^{|j|}, e^{ij\theta}) \mapsto (R^{|j|}, e^{ij\Theta})$, $j \in \mathbb Z$, given by the Weyl transform. In the final section we comment on connections between local decay estimates for operators in the pre=image and image of the Weyl transform.

	\section{Notation and conventions}

	We take $\mathbb Z_+ = \{0, 1, \ldots\}$, $\mathbb R_+ = [0, \infty)$. Let \emph{commutative phase space}, $\mathscr C$, be the commutative Hilbert algebra of $L^2(\mathbb R^2, \mathbb C, \mathrm dx\mathrm dy)$ functions, where $(x, y)$ are distinguished coordinates on $\mathbb R^2$ and $\mathrm dx\mathrm dy$ is the Lebesgue measure on $\mathrm R^2$, with pointwise multiplication $(f, g) \mapsto fg$ such that $(fg)(x, y) = f(x, y)g(x, y)$. We denote the inner product on $\mathscr C$ by
	\begin{align}
	\langle f, g\rangle &= \int_{\mathbb R^2}\mathrm dx\mathrm dy\ \overline f(x, y)g(x, y),\quad \forall f, g \in \mathscr C,
	\end{align}
	where $\zeta \mapsto \overline \zeta$ is complex conjugation for all $\zeta \in \mathbb C$ and extends to an action on functions pointwise.
	
	We will always view $\mathscr C$ not just as a Hilbert space but as an algebra so we will refer to the appropriate operators thereupon as automorphisms and derivations. We denote, for each $f \in \mathscr C$, $m_f$ to be the multiplication automorphism which formally acts as $g \mapsto fg$ for $g \in \mathscr C$. We denote the coordinates of the domain of functions of $\mathscr C$ by $(x, y) \in \mathbb R^2$, and distinguished 1D subspaces by the coordinates $x, y \in \mathbb R$. By abuse of notation we take the identity functions on $x, y \in \mathbb R$ to be denoted by $x, y$ and so that $m_xf(x, y) = xf (x, y)$ and $m_yf(x, y) = yf(x, y)$ for all suitable $f \in \mathscr C$. We take by abuse of notation $p_x, p_y$ to be the momentum derivations on $\mathscr C$ that act as
	\begin{align}
	p_xf(x) = -i\epsilon\partial_xf(x, y),\quad p_yf(x, y) = -i\epsilon\partial_yf(x, y)
	\end{align}
	on all suitable $f \in \mathscr C$ for fixed $\epsilon \in \mathbb R_+\setminus\{0\}$. We take the coordinates of the domain of $\mathscr F\mathscr C$ to be denoted by $(\xi_x, \xi_y) \in \mathbb R^2$, where $\mathscr F$ is the Fourier transform
	\begin{align}
	\mathscr Ff(\xi_x, \xi_y) = (2\pi)^{-1}\int_{\mathbb R^2}\mathrm dx\mathrm dy\ e^{-i(\xi_xx +\xi_yy)}f(x, y).
	\end{align}

	Let \emph{noncommutative phase space}, $\mathscr N$, be the noncommutative Hilbert algebra of Hilbert-Schmidt operators on $\mathscr H = L^2(\mathbb R, \mathbb C, \mathrm dx)$, where $x$ is a distinguished coordinate on $\mathbb R$ and $\mathrm dx$ is the Lebesgue measure on $\mathrm R$. Let $\{\chi_n\}_{n = 0}^\infty$ be a basis of $\mathscr H$, such that $\phi = \sum_{n = 0}^\infty\phi_n\chi_n$. One may then write the formal action of $A \in \mathscr N$ on a $\phi \in \mathscr H$ as
	\begin{align}
	A\phi(x) &= \int_\mathbb R\mathrm dx'\ A(x, x')\phi(x'),\quad A\phi_n = \sum_{n' = 0}^\infty A_{n, n'}\phi_{n'}.
	\end{align}
	We denote the inner product on $\mathscr N$ by
	\begin{align}
	\langle A, B\rangle &= \mathrm{tr}(A^\dag B) = \int_{\mathbb R^2}\mathrm dx\mathrm dx'\ \overline A(x, x')B(x, x') = \sum_{n, n' = 0}^\infty\overline A_{n, n'}B_{n, n'}
	\end{align}
	for all $A, B \in \mathscr N$.
	
	We will always view $\mathscr N$ not just as a Hilbert space but as an algebra so we will refer to the appropriate operators thereupon as automorphisms and derivations. We denote by $M^\mathrm L_A, M^\mathrm R_A$, for $A \in \mathscr N$ of suitable regularity, to be respectively the left and right multiplication automorphisms which act as
	\begin{align}
	M^\mathrm L_AB = AB,\quad M^\mathrm R_AB = BA.
	\end{align}
	We denote by $M_A$ the multiplication automorphism which acts as
	\begin{align}
	M_AB = (AB +BA)/2.
	\end{align}
	We write $\mathrm{ad}_A$ for the adjoint derivation which acts as $B \mapsto [A, B]$, where $[A, B] = AB -BA$ is the commutator. We take $\mathrm{Ad}_A$ to be the adjoint group automorphism that acts as $B \mapsto ABA^{-1}$. Let $X, Y \in \mathscr C$ be the coordinate operators which formally act as
	\begin{align}
	X\phi(x) = x\psi(x),\quad Y\phi(x) = -i\epsilon\partial_x\psi(x),
	\end{align}
	for all sufficiently regular $\psi \in \mathscr H$ for fixed $\epsilon \in \mathbb R_+\setminus\{0\}$, and $P_X, P_Y$ the momentum derivations which formally act as
	\begin{align}
	P_XA = \mathrm{ad}_YA,\quad P_YA = -\mathrm{ad}_XA
	\end{align}
	for all $A \in \mathscr C$. We also write $(X, Y) \in \mathscr C^2$.

Let $L^1(\mathbb R^2) \equiv L^1(\mathbb R^2, \mathbb C, \mathrm dx\mathrm dy)$. We denote by $\mathscr S(\mathbb R^2)$ and $\mathscr S(\mathbb R)$ respectively the Schwartz class functions of $\mathscr C$ and $\mathscr H$. We take $\mathscr S^* (\mathbb R^2)$ and $\mathscr S^*(\mathbb R)$ to be the spaces of tempered distributions dual to the associated Schwartz spaces. All extra structure associated with $\mathscr C$ and $\mathscr H$ will be extended to $\mathscr S(\mathbb R^2)$, $\mathscr S(\mathbb R)$, and their duals where applicable. If $f(\cdot, \cdot) \equiv f: \mathbb R^2 \to \mathbb C$ is any complex function on $\mathbb R^2$ with cartesian representation $f: (x, y) \mapsto f(x, y)$ then we take $f'(\cdot, \cdot) \equiv f': \mathbb R^2 \to \mathbb C$ to be the complex function with polar representation $f': (\rho, \theta) \mapsto f'(\rho, \theta)$, where $(\rho, \theta)$ are polar coordinates that satisfy $(x, y) = (\rho\cos\theta, \rho\sin\theta)$. We likewise use the same notational distinction for the Fourier transform of suitably regular functions via $\mathscr Ff: (\xi_x, \xi_y) \mapsto \mathscr Ff(\xi_x, \xi_y)$ and $\mathscr Ff': (\varrho, \vartheta) \mapsto \mathscr Ff'(\varrho, \vartheta)$, where $(\xi_x, \xi_y) = (\varrho\cos\vartheta, \varrho\sin\vartheta)$.

	\section{Overview of the Weyl transform}
	
	\begin{defn}
	We define $\mathscr W$ to be \emph{Weyl transform} formally specified by
	\begin{align}
	\mathscr Wf = (2\pi)^{-1}\int_{\mathbb R^2}\mathrm d\xi_x\mathrm d\xi_y\ e^{i(\xi_xX +\xi_yY)} \mathscr Ff(\xi_x, \xi_y),
	\end{align}
	where $\mathscr F$ is the \emph{Fourier transform}:
	\begin{align}
	\mathscr Ff(\xi_x, \xi_y) = (2\pi)^{-1}\int_{\mathbb R^2}\mathrm dx\mathrm dy\ e^{-i(\xi_xx +\xi_yy)}f(x, y).
	\end{align}
	\end{defn}
	
	The Weyl transform is manifestly linear so one is free to write $\mathscr W(f) \equiv \mathscr Wf$ without ambiguity. The domain and image of the operator remains to be specified. The elements of $\mathscr C$, $f(r) \equiv f(x, y)$, and the integral kernels of linear operators on $\mathscr H$, $A (x, x')$, are similar in form. Therefore, it should be expected that there exist constraints on both spaces which render them manifestly isomorphic. This points to a number of common rigorous definitions of $\mathscr W$:
	\begin{enumerate}
	\item Bijection from $\mathscr V = \{f: \mathscr Ff \in L^1(\mathbb R^2)\}$ to the space of bounded linear operators on $\mathscr H$, reviewed in \cite{Fo89}.
	\item Bijection from $\mathscr S^*(\mathbb R^2)$ to the space of continuous linear maps from $\mathscr S(\mathbb R)$ to $\mathscr S^*(\mathbb R)$, reviewed in \cite{Fo89}.
	\item Unitary bijection from $\mathscr C$ to $\mathscr N$, shown in \cite{Po66}\cite{Se63}.
	\item Injection of $\mathscr S(\mathbb R^2)$ into the space of trace-class operators on $\mathscr H$, shown in \cite{Lo66}\cite{LoMi67}.
	\item Injection of $L^1(\mathbb R^2)$ into the space of compact operators on $\mathscr H$, shown in \cite{Lo66}\cite{LoMi67}.
	\end{enumerate}
	
	Much of the fundamental theory of the Weyl transform, as well as numerous enlightening observations, is thoroughly reviewed by Folland in \cite{Fo89}. Most of our review of the properties of $\mathscr W$ is a direct summary thereof. We direct the reader to this excellent resource for further details. A host of additional results, with an eye for counterintuitive aspects, are explored by Daubechies in \cite{Da83}. We consider this a helpful complement to Folland's monograph \cite{Fo89}.
	
	The motivation behind the definition of $\mathscr W$ is as follows. Consider the coordinate functions $x, y$ as tempered distributions. The Weyl transform gives $\mathscr Wx = X, \mathscr Wy = Y$. There are many similar transforms which provide this correspondence. This family of transforms can be parametrized through the replacement $W_{\xi_x, \xi_y} \mapsto \alpha(\xi_x, \xi_y)W_{\xi_x, \xi_y}$ for some function $(\xi_x, \xi_y) \mapsto \alpha(\xi_x, \xi_y) \in \mathbb C$. The Weyl transform is the transform, unique up to a multiplicative scalar factor, that maps commutative functions to noncommutative functions with symmetric ordering in the following sense. If $f(x, y) = x^my^n$, for some $m, n \in \mathbb Z_+$, then $\mathscr Wf$ is the unique polynomial in $X, Y$ that is invariant under $(X, Y) \mapsto (Y, X)$ and reduces to $f(x, y)$ under the substitution $(X, Y) \mapsto (x, y)$.
	
	In the above sense one may consider $\mathscr C$ as the classical phase space of a point particle moving in 1D and $\mathscr H$ as the associated quantum state space induced by a suitable choice of quantization. The Weyl transform gives a distinguished quantization by essentially sending polynomial functions of $(x, y)$ to symmetrically ordered polynomial functions of $(X, Y)$. Indeed, it is worth noting that the Weyl transform is typically used to present a means of ``quantizing'' a classical system, e.g. transforming Hamilton's equations of motion into Heisenberg's equations of motion. This method is often called \emph{deformation quantization}, as the Weyl transform effectively deforms the commutative algebra $\mathscr C$ into the noncommutative algebra $\mathscr N$. The fundamental relation that determines the structure of $\mathscr N$ is that of $[X, Y] = i\epsilon I$ and therefore one can ``undo'' the deformation by taking $\epsilon \to 0$ instead of applying $\mathscr W^{-1}$. This singular limit is typically the method considered in physics.
	
	 The Weyl transformation furnishes a noncommutative product on $\mathscr C$ specified by $f\star g = \mathscr W^{-1}(\mathscr Wf\mathscr Wg)$. This product is typically called the \emph{Moyal product}. By introducing this product one deforms, i.e. fundamentally changes, the algebraic structure on $\mathscr C$. In this sense the Weyl transform can also be used to replace automorphisms with new ones. The most common way to do this is to first take an automorphism which multiplies by a function and replace it by the function itself. One then deforms the pointwise multiplication of a potential automorphism with a Moyal product: $vf \mapsto v\star f$ or $vf \mapsto f\star v$, where $v, f \in \mathcal C$. By applying a forward Weyl transform one can see that this is equivalent to transforming the two functions separately and then multiplying them: $\mathscr W(v\star f) = VF$ or $\mathscr W(f\star v) = FV$, where $v, f \in \mathscr C$, $V = \mathscr Wv , F = \mathscr Wf$.

	The deformation picture should be carefully distinguished from that of the case where $v$ is an automorphism on $\mathscr C$ and should be transformed as $v \mapsto \mathscr Wv\mathscr W^{-1}$. This situation with $(\mathscr W, \star)$ is analogous to that with $(\mathscr F, *)$. This is to say, that one can use the properties of the Fourier transform to study an operator created by deforming a pointwise multiplication operator into a convolution operator.

	We will not explore deformations in this work and instead will focus on constructively elucidating the action of the Weyl transform as a unitary bijection. We recall that the Weyl transform as a map from $L^2(\mathbb R^2)$ to Hilbert-Schmidt operators on $L^2(\mathbb R)$ is a unitary bijection. Therefore $\mathscr W$ will leave the spectra of operators invariant and provide possibly simpler representations with which to study these spectra. These properties can be used to furnish a number of related results of interest.

	\section{Elements of analysis on noncommutative space}
	
	First one may formally observe that since
	\begin{align}
	&X\psi = x\psi,\quad Y\psi = -i\epsilon\partial_x\psi,
	\end{align}
	it is the case that
	\begin{align}
	&[X, Y]\psi = (XY -YX)\psi = x(-i\epsilon\partial_x)\psi -(-i\epsilon\partial_x)(x\psi)\\
	&\qquad = x(-i\epsilon\partial_x)\psi -[(-i\epsilon\partial_x)(x)\psi +x(-i\epsilon\partial_x)\psi]\\
	&\qquad = i\epsilon\psi,\\
	&[X, Y] = i\epsilon I.
	\end{align}
	Historically this is backward as one would first take the relation $[X, Y] = i\epsilon I$, then observe that this is satisfied by the representation $X\psi = x\psi,\quad Y\psi = -i\epsilon\partial_x\psi$, and then by Stone-von Neumann theorem determine that this representation is unique up to unitary isomorphism. It is the case that, however, each value of $\epsilon \in \mathbb R_+\setminus\{0\}$ (actually $\epsilon \in \mathbb R\setminus\{0\}$ one has a unitarily inequivalent representation. At the level of the generators alone, one may observe how the commutator relation is transformed under translations, rotations, and dilations. If one takes $(X, Y) \mapsto (X', Y')$ with
	\begin{align}
	X' = c\cos\theta X -c\sin\theta Y +aI,\quad Y' = c\sin\theta X +c\cos\theta Y +Ib,\quad a, b, c \in \mathbb R.
	\end{align}
	then one has
	\begin{align}
	[X', Y'] &= [c\cos\theta X -c\sin\theta Y +aI, c\sin\theta X +c\cos\theta Y +bI]\\
	&= c^2\cos(\theta)^2[X, Y] -c^2\sin(\theta)^2[Y, X] = c^2[\cos(\theta)^2 +\sin(\theta)^2][X, Y]\\
	&= c^2[X, Y] = ic^2\epsilon I.
	\end{align}
	One can see then that, up to a global choice of scale, the commutator relation is preserved. Moreover the parameter $\epsilon$ can be viewed as encoding a kind of global scale. The commutative limit, $\epsilon \to 0$, can then be viewed as equivalent to a ``zooming out'' in the sense of scale. One can then expect that there may be correspondences between formulae for $\mathscr C$ and $\mathscr N$ in the spatially asymptotic sense. This can often be the case.

	Since $[X, Y] = i\epsilon I$, the Campbell-Baker-Hausdorf formula for the operators $X, Y$ gives
	\begin{align}
	e^{i\xi_xX}e^{i\xi_yY} &= e^{i\xi_xX +i\xi_yY +[i\xi_xX, i\xi_yY]/2} = e^{i\xi_xX +i\xi_yY -\xi_x\xi_y[X, Y]/2}\\
	&= e^{i\xi_xX +i\xi_yY -i\epsilon\xi_x\xi_yI/2} = e^{-i\epsilon\xi_x\xi_y/2}e^{i(\xi_xX +\xi_yY)}\\
	e^{i(\xi_xX +\xi_yY)} &= e^{i\epsilon\xi_x\xi_y/2}e^{i\xi_xX}e^{i\xi_yY} = e^{-i\epsilon\xi_x\xi_y/2}e^{i\xi_yY}e^{i\xi_xX}
	\end{align}
	From which follows
	\begin{align}
	\mathrm{Ad}_{\exp(i\xi_xX)}e^{i\xi_yY} &= e^{i\xi_xX}e^{i\xi_yY}e^{-i\xi_xX} = e^{i\xi_y(Y -\epsilon\xi_x)},\\
	\mathrm{Ad}_{\exp(i\xi_yY)}e^{i\xi_xX} &= e^{i\xi_yY}e^{i\xi_xX}e^{-i\xi_yY} = e^{i\xi_x(X +\epsilon\xi_y)},
	\end{align}
	and formally
	\begin{align}
	i\epsilon\xi_xe^{i\xi_xX} &= \partial_ae^{i\xi_x(X +\epsilon a)}\downharpoonright_{a = 0} = i\mathrm{ad}_Ye^{i\xi_xX},\\
	i\epsilon\xi_ye^{i\xi_yY} &= \partial_ae^{i\xi_y(Y +\epsilon a)}\downharpoonright_{a = 0} = -i\mathrm{ad}_Xe^{i\xi_yY}.
	\end{align}
	We will also use the formal identities
	\begin{align}
	-i\partial_{\xi_x}e^{i\xi_xX} &= Xe^{i\xi_xX},\\
	-i\partial_{\xi_y}e^{i\xi_yY} &= Ye^{i\xi_yY}.
	\end{align}

	\begin{defn}
	Let $W_{\xi_x, \xi_y}$ be the \emph{Weyl operator} formally specified by
	\begin{align}
	W_{\xi_x, \xi_y} := e^{i(\xi_xX +\xi_y Y)},\quad \xi_x, \xi_y \in \mathbb R.
	\end{align}
	\end{defn}
	
	The above definition allows one to write the transform as
	\begin{align}
	\mathscr Wf = (2\pi)^{-1}\int_{\mathbb R^2}\mathrm d\xi_x\mathrm d\xi_y\ W_{\xi_x, \xi_y}\mathscr Ff(\xi_x, \xi_y),
	\end{align}
	and the inverse transform as
	\begin{align}
	\mathscr W^{-1}A(x, y) = (2\pi)^{-1}\int_{\mathbb R^2}\mathrm d\xi_x\mathrm d\xi_y\ e^{i(\xi_xx +\xi_yy)}\mathrm{tr}(W_{\xi_x, \xi_y}^\dag A).
	\end{align}
	In general it can be difficult to determine the image of the inverse transform explicitly. Fortunately there is a straightforward correspondence property for rank 1 operators.
	
	\begin{lem}
	It is the case that
	\begin{align}
	\langle\phi, W_{\xi_x, \xi_y}\psi\rangle = \int_\mathbb R\mathrm dx\ e^{i\xi_xx}\overline\phi(x -\epsilon\xi_y/2)\psi(x +\epsilon\xi_y/2).
	\end{align}
	\end{lem}
	
	\begin{proof}
	One has, for sufficiently regular $\phi, \psi \in \mathscr H$, that
	\begin{align}
	W_{\xi_x, \xi_y}\phi(x) &= e^{i\epsilon\xi_x\xi_y/2}e^{i\xi_xX}e^{i\xi_yY}\phi(x) = e^{i\epsilon\xi_x\xi_y/2}e^{i\xi_xX}e^{\epsilon\xi_y\partial_x}\phi(x)\\
	&= e^{i\epsilon\xi_x\xi_y/2}e^{i\xi_xX}\phi(x +\epsilon\xi_y) = e^{i\epsilon\xi_x\xi_y/2}e^{i\xi_xx}\phi(x +\epsilon\xi_y)\\
	&= e^{i\epsilon\xi_x\xi_y/2 +i\xi_xx}\phi(x +\epsilon\xi_y) = e^{i\xi_x(x +\epsilon\xi_y/2)}\phi(x +\epsilon\xi_y),
	\end{align}
	and therefore
	\begin{align}
	\langle\phi, W_{\xi_x, \xi_y}\psi\rangle &= \int_\mathbb R\mathrm dx\ e^{i\xi_x(x +\epsilon\xi_y/2)}\overline\phi(x)\psi(x +\epsilon\xi_y)\\
	&= \int_\mathbb R\mathrm dx\ e^{i\xi_xx}\overline\phi(x -\epsilon\xi_y/2)\psi(x +\epsilon\xi_y/2).
	\end{align}
	\end{proof}
	
	\begin{prop}
	For all Schwartz class $\phi, \psi \in \mathscr H$ it is the case that
	\begin{align}
	\mathscr W^{-1}(\psi\otimes\phi^*)(x, y) &= 2\epsilon^{-1}\langle\phi, W_{2\epsilon^{-1}y, -2\epsilon^{-1}x}\widehat\psi\rangle,
	\end{align}
	where $\widehat\psi(x) = \psi(-x)$.
	\end{prop}
	
	\begin{proof}
	Consider the operator $A = \psi\otimes\phi^*$. One has that
	\begin{align}
	\mathscr W^{-1}A(x, y) &= (2\pi)^{-1}\int_{\mathbb R^2}\mathrm d\xi_x\mathrm d\xi_y\ e^{i(\xi_xx +\xi_yy)}\mathrm{tr}(W_{\xi_x, \xi_y}^\dag A)\\
	&= (2\pi)^{-1}\int_{\mathbb R^2}\mathrm d\xi_x\mathrm d\xi_y\ e^{i(\xi_xx +\xi_yy)}\mathrm{tr}(W_{-\xi_x, -\xi_y}A)\\
	&= (2\pi)^{-1}\int_{\mathbb R^2}\mathrm d\xi_x\mathrm d\xi_y\ e^{i(\xi_xx +\xi_yy)}\langle\phi, W_{-\xi_x, -\xi_y}\psi\rangle\\
	&= (2\pi)^{-1}\int_{\mathbb R^2}\mathrm d\xi_x\mathrm d\xi_y\ e^{i(\xi_xx +\xi_yy)}\\
	&\qquad \times\int_\mathbb R\mathrm dx'\ e^{-i\xi_xx'}\overline\phi(x' +\epsilon\xi_y/2)\psi(x' -\epsilon\xi_y/2).
	\end{align}
	
	Then by Fourier inversion one has
	\begin{align}
	\mathscr W^{-1}A(x, y) &= \int_\mathbb R\mathrm d\xi_y\ e^{i\xi_yy}\overline\phi(x +\epsilon\xi_y/2)\psi(x -\epsilon\xi_y/2)\\
	&= (2/\epsilon)\int_\mathbb R\mathrm d\xi_y\ e^{i\xi_y2\epsilon^{-1}y}\overline\phi(\xi_y +x)\psi(-\xi_y +x)\\
	&= (2/\epsilon)\int_\mathbb R\mathrm d\xi_y\ e^{i\xi_y2\epsilon^{-1}y}\overline\phi(\xi_y +x)\widehat\psi(\xi_y -x)\\
	&= (2/\epsilon)\int_\mathbb R\mathrm d\xi_y\ e^{i2\epsilon^{-1}y\xi_y}\overline\phi(\xi_y -\epsilon\{-2\epsilon^{-1}x\}/2)\widehat\psi(\xi_y +\epsilon\{-2\epsilon^{-1}x\}/2),
	\end{align}
	where $\widehat\psi(x) = \psi(-x)$. Thus
	\begin{align}
	\mathscr W^{-1}(\psi\otimes\phi^*)(x, y) &= 2\epsilon^{-1}\langle\phi, W_{2\epsilon^{-1}y, -2\epsilon^{-1}x}\widehat\psi\rangle.
	\end{align}
	\end{proof}
	
	\begin{rem}
	By taking the closure of Schwartz class functions in $\mathscr H$ the above correspondence property extends to all $\phi, \psi \in \mathscr H$.
	\end{rem}
	
	\begin{prop}
	It is the case that
	\begin{align}
	\mathscr W(p_xf) &= P_X\mathscr Wf,\quad \mathscr W(p_yf) = P_Y\mathscr Wf.
	\end{align}
	\end{prop}
	
	\begin{proof}
	One has on Schwartz class functions that
	\begin{align}
	\mathscr W(-i\epsilon\partial_xf) &= (2\pi)^{-1}\int_{\mathbb R^2}\mathrm d\xi_x\mathrm d\xi_y\ e^{i(\xi_xX +\xi_yY)}\mathscr F(-i\epsilon\partial_xf)(\xi_x, \xi_y)\\
	&= (2\pi)^{-1}\int_{\mathbb R^2}\mathrm d\xi_x\mathrm d\xi_y\ e^{i(\xi_xX +\xi_yY)}\epsilon\xi_x\mathscr Ff(\xi_x, \xi_y)\\
	&= (2\pi)^{-1}\int_{\mathbb R^2}\mathrm d\xi_x\mathrm d\xi_y\ \epsilon\xi_xe^{i(\xi_xX +\xi_yY)}\mathscr Ff(\xi_x, \xi_y )\\
	&= (2\pi)^{-1}\int_{\mathbb R^2}\mathrm d\xi_x\mathrm d\xi_y\ \epsilon\xi_xe^{i\xi_x(X +\epsilon\xi_y/2)}e^{i\xi_yY}\mathscr Ff(\xi_x, \xi_y)\\
	&= (2\pi)^{-1}\int_{\mathbb R^2}\mathrm d\xi_x\mathrm d\xi_y\ \mathrm{ad}_Ye^{i\xi_x(X +\epsilon\xi_y/2)}e^{i\xi_yY}\mathscr Ff(\xi_x, \xi_y)\\
	&= \mathrm{ad}_Y(2\pi)^{-1}\int_{\mathbb R^2}\mathrm d\xi_x\mathrm d\xi_y\ e^{i\xi_x(X +\epsilon\xi_y/2)}e^{i\xi_yY}\mathscr Ff(\xi_x, \xi_y)\\
	&= \mathrm{ad}_Y\mathscr Wf.
	\end{align}
	A similar calculation holds for $\mathscr W(-i\epsilon\partial_yf)$ which yields a result that differs by an overall sign.
	\end{proof}

	The above correspondence holds for all transforms related to the Weyl one through the replacement $W_{\xi_x, \xi_y} \mapsto \alpha(\xi_x, \xi_y)W_{\xi_x, \xi_y}$. There is an important correspondence which holds for alone for transforms related to the Weyl transform through an overall multiplicative scalar factor.
	\begin{prop}
	It is the case that $\mathscr W(m_xf) = M_X\mathscr Wf, \mathscr W(m_yf) = M_Y\mathscr Wf$.
	\end{prop}
	
	\begin{proof}
	On Schwartz class functions one may find
	\begin{align}
	\mathscr W(m_xf) &= (2\pi)^{-1}\int_{\mathbb R^2}\mathrm d\xi_x\mathrm d\xi_y\ e^{i(\xi_xX +\xi_yY)}\mathscr F[xf(\xi_x, \xi_y)]\\
	&= (2\pi)^{-1}\int_{\mathbb R^2}\mathrm d\xi_x\mathrm d\xi_y\ e^{i(\xi_xX +\xi_yY)}i\partial_{\xi_x}\mathscr Ff(\xi_x, \xi_y)\\
	&= (2\pi)^{-1}\int_{\mathbb R^2}\mathrm d\xi_x\mathrm d\xi_y\ (-1)i\partial_{\xi_x}e^{i(\xi_xX +\xi_yY)}\mathscr Ff(\xi_x, \xi_y)\\
	&= (2\pi)^{-1}\int_{\mathbb R^2}\mathrm d\xi_x\mathrm d\xi_y\ (-1)i\partial_{\xi_x}e^{i\xi_x(X +2^{-1}\epsilon\xi_y)}e^{i\xi_yY}\mathscr Ff(\xi_x, \xi_y)\\
	&= (2\pi)^{-1}\int_{\mathbb R^2}\mathrm d\xi_x\mathrm d\xi_y\ (X +2^{-1}\epsilon\xi_y)e^{\xi_x(X +2^{-1}\epsilon\xi_y)}e^{i\xi_yY}\mathscr Ff(\xi_x, \xi_y)\\
	&= (2\pi)^{-1}\int_{\mathbb R^2}\mathrm d\xi_x\mathrm d\xi_y\ (M^\mathrm L_X -2^{-1}\mathrm{ad}_X)e^{i(\xi_xX +\xi_yY)}\mathscr Ff(\xi_x, \xi_y)\\
	&= (M^\mathrm L_X -2^{-1}\mathrm{ad}_X)\mathscr Wf\\
	&= M_X\mathscr Wf.
	\end{align}
	A similar result holds for $m_y$.
	\end{proof}

	\section{Representations in terms of special bases}

	Application of the Weyl transform can be difficult without using explicit bases in all spaces. We will demonstrate the utility of this perspective by continuing to follow the methods outlined in \cite{Fo89}.
	
	\begin{defn}
	We denote by $H_n : \mathbb R \to \mathbb R$, $n \in \mathbb Z_+$, $H_n: x \mapsto H_n(x)$, the \emph{Hermite polynomials}, defined by \cite[\S~18.3]{DLMF}:
	\begin{align}
	&H_n(x) = n!\sum_{j = 0}^{\lfloor n/2\rfloor} \frac{(-1)^j}{j!(n -2j)!}(2x)^{n -2j},\quad H_0(x) = 1.
	\end{align}
	We define $\mathcal A$ to be the discrete operator formally defined by
	\begin{align}
	&\mathcal Av_n := 2^{-1}v_{n +1} +nv_{n -1},
	\end{align}
	for all $v: \mathbb Z_+ \to \mathbb C$, $v: n \mapsto v_n$.
	\end{defn}
	
	\begin{rem}
	The Hermite polynomials, $H_n$, satisfy \cite[\S~18.3]{DLMF}:
	\begin{align}
	&H_n(-x) = (-1)^nH_n(x),\\
	&\sum_{n = 0}^\infty H_n(x) \frac{\zeta^n}{n!} = \exp(2x s -s^2),\quad \zeta\in\mathbb C,\\
	&\int_\mathbb R \mathrm dx\ e^{-x^2} H_m(x)H_n(x) = \pi^{1/2}2^mm!\delta_{m, n},\\
	&\mathcal AH_n(x) = xH_n(x).
	\end{align}
	The operator $\mathcal A$ is singular at $n = 0$ so boundary conditions needn't be specified.
	\end{rem}
	
	\begin{defn}
	We denote by $h_n : \mathbb R \to \mathbb R$, $n \in \mathbb Z_+$, $h_n: x \mapsto h_n(x)$, the \emph{parametrized Hermite functions}. They satisfy
	\begin{align}
	&h_n(x) := \epsilon^{-1/4}\pi^{-1/4}2^{-n/2}(n!)^{-1/2}e^{-x^2/2\epsilon}H_n(\epsilon^{-1/2}x),\quad \epsilon \in \mathbb R\setminus\{0\}.
	\end{align}
	Define $\mathcal B$ to be discrete operator formally defined by
	\begin{align}
	&\mathcal Bv_n := 2^{-1/2}(n +1)^{1/2}v_{n +1} +2^{-1/2}n^{1/2}v_{n -1},
	\end{align}
	for all $v: \mathbb Z_+ \to \mathbb C$, $v: n \mapsto v_n$. Let the \emph{parametrized Hermite operators} be given by
	\begin{align}
	E_{m, n} := h_m\otimes h^*_n.
	\end{align}
	\end{defn}
	
	\begin{rem}
	The parametrized Hermite functions, $h_n$, satisfy:
	\begin{align}
	&\langle h_m, h_n\rangle = \int_\mathbb R\mathrm dx\ h_m(x)h_n(x) = \delta_{m, n},\\
	&\mathcal Bh_n(x) = \epsilon^{-1/2}xh_n(x).
	\end{align}
	The operator $\mathcal B$ is manifestly formally symmetric, furthermore it is singular so boundary conditions needn't be specified. The $h_n$ also satisfy
	\begin{align}
	\sum_{n = 0}^\infty h_n(x)h_n(y) &= \delta(x -y)
	\end{align}
	in the distributional sense in $L^2(\mathbb R)$. The parametrized Hermite operators, $\{E_{m, n}\}_{m, n \in \mathbb Z_+}$, form an orthonormal basis of $\mathscr N$ for all $\epsilon$.
	\end{rem}
	
	\begin{defn}
	We denote by $\{\kappa_n\}_{n = 0}^\infty$, an orthonormal basis of the \emph{Bargmann-Segal} (BS) space of holomorphic functions $f \in L^2(\mathbb C, \mathbb C, w(s)\mathrm d^2s)$, specified by
	\begin{align}
	&\kappa_n(s) = \epsilon^{-1/4}\pi^{-1/2}(n!)^{-1/2}(\epsilon^{-1/2}s)^n,\quad w(s) = e^{-|s|^2/\epsilon^{-1} },\quad \epsilon \in \mathbb R\setminus\{0\}.
	\end{align}
	\end{defn}
	
	\begin{rem}
	One may observe that
	\begin{align}
	&\int_\mathbb C\mathrm d^2s\ w(s)\overline\kappa_m(s)\kappa_n(s) = \delta_{m, n},
	\end{align}
	\begin{align}
	\sum_{n = 0}^\infty\overline\kappa_n(s_1)\kappa_n(s_2) &= \sum_{n = 0}^\infty\epsilon^{-1/2}\pi^{-1}(n!)^{-1}(\overline s_1s_2/\epsilon)^n = \epsilon^{-1/2}\pi^{-1}e^{\overline s_1s_2/\epsilon}.
	\end{align}
	\end{rem}
	
	\begin{prop}
	There exists a linear unitary isomorphism $\mathscr B$ that maps $\mathscr H$ to the BS space such that
	\begin{align}
	\mathscr B h_n = \kappa_n,
	\end{align}
	for any fixed $\epsilon \in \mathbb R\setminus\{0\}$, where
	\begin{align}
	\mathscr Bf(s) = \int_\mathbb R\mathrm dx\ \mathscr B(s, x)f(x),\quad \mathscr B^{-1}g(x) = \int_\mathbb C\mathrm d^2s\ \mathscr B^{-1}(x, s)g(s),
	\end{align}
	for all $f \in \mathscr H$ and all $g$ in the BS space.
	\end{prop}
	
	\begin{proof}
	One may find an explicit representation of $\mathscr B$ through the exponential generating function of the usual Hermite polynomials and the orthonormality of the two bases:
	\begin{align}
	\sum_{n = 0}^\infty(n!)^{-1}s^nH_n(x) &= \exp(-s^2 +2sx).
	\end{align}
	We then find
	\begin{align}
	\mathscr B(s, x) &= \sum_{n = 0}^\infty\kappa_n(s)h_n(x)\\
	&= \epsilon^{-1/2}\pi^{-3/4}\exp(-2^{-1}\epsilon^{-1}s^2 +2^{1/2}\epsilon^{-1}sx -2^{-1}\epsilon^{-1}x^2),\\
	\mathscr B^{-1}(x, s) &= \sum_{n = 0}^\infty h_n(x)w(s)\overline\kappa_n(s)\\
	&= \epsilon^{-1/2}\pi^{-3/4}w(s)\exp(-2^{-1}\epsilon^{-1}\overline s^2 +2^{1/2}\epsilon^{-1}\overline sx -2^{-1}\epsilon^{-1}x^2),\\
	&= w(s)\mathscr B(\overline s, x).
	\end{align}
	These explicit constructions manifestly satisfy the requirements of the given definition.
	\end{proof}
	
	\begin{defn}
	Let the \emph{lowering and raising operators}, $\mathcal S, \mathcal S^\dag \in \mathscr N$, be formally given respectively by
	\begin{align}
	&\mathcal S := (2\epsilon)^{-1/2}(X +iY),\quad \mathcal S^\dag := (2\epsilon)^{-1/2}(X -iY).
	\end{align}
	\end{defn}
	
	\begin{prop}
	It is the case that the map $\mathscr B$ allows one to express the action of the raising and lowering operators as
	\begin{align}
	&\mathscr B\mathcal S^\dag = \epsilon^{-1/2}s\mathscr B,\quad \mathscr B\mathcal S = \epsilon^{1/2}\partial_s \mathscr B.
	\end{align}
	\end{prop}
	
	\begin{cor}
	\begin{align}
	&\mathcal S^\dag h_n = (n +1)^{1/2}h_{n +1},\quad \mathcal S h_n = n^{1/2}h_{n -1}.
	\end{align}
	\end{cor}
	
	\begin{proof}
	First we observe that
	\begin{align}
	\epsilon^{-1/2}s\kappa_n(s) = \kappa_{n +1}(s),\quad \epsilon^{1/2}\partial_s\kappa_{n +1}(s) = \kappa_n.
	\end{align}
	One may find
	\begin{align}
	\mathscr B(s, x) &= \epsilon^{-1/2}\pi^{-3/4}\exp(-2^{-1}\epsilon^{-1}s^2 +2^{1/2}\epsilon^{-1}sx -2^{-1}\epsilon^{-1}x^2),\\
	\epsilon^{1/2}\partial_s\mathscr B(s, x) &= (-\epsilon^{-1/2}s +2^{1/2}\epsilon^{-1/2}x)\mathscr B(s, x),\\
	2^{1/2}\epsilon^{-1/2}x\mathscr B(s, x) &= (\epsilon^{-1/2}s +\epsilon^{1/2}\partial_s)\mathscr B(s, x),\\
	\epsilon^{-1/2}x\mathscr B(s, x) &= 2^{-1/2}(\epsilon^{-1/2}s +\epsilon^{1/2}\partial_s)\mathscr B(s, x)\\
	&= (2^{-1/2}\epsilon^{-1/2}s +2^{-1/2}\epsilon^{1/2}\partial_s)\mathscr B(s, x),\\
	\epsilon^{1/2}\partial_x\mathscr B(s, x) &= (2^{1/2}\epsilon^{-1/2}s -\epsilon^{-1/2}x)\mathscr B(s, x)\\
	&= (2^{1/2}\epsilon^{-1/2}s -2^{-1/2}\epsilon^{-1/2}s -2^{-1/2}\epsilon^{1/2}\partial_s)\mathscr B(s, x),\\
	&= (2^{-1/2}\epsilon^{-1/2}s -2^{-1/2}\epsilon^{1/2}\partial_s)\mathscr B(s, x).
	\end{align}
	\begin{align}
	\mathscr B\psi(s) &= \int_\mathbb R\mathrm dx\ \mathscr B(s, x)\psi(x),\\
	\mathscr B\epsilon^{-1/2}X\psi(s) &= \int_\mathbb R\mathrm dx\ \mathscr B(s, x)\epsilon^{-1/2}X\psi(x)\\
	&= \int_\mathbb R\mathrm dx\ \mathscr B(s, x)\epsilon x\psi(x)\\
	&= \int_\mathbb R\mathrm dx\ \epsilon^{-1/2}x\mathscr B(s, x)\psi(x),\\
	&= (2^{-1/2}\epsilon^{-1/2}s +2^{-1/2}\epsilon^{1/2}\partial_s)\mathscr B\psi(s),
	\end{align}
	\begin{align}
	\mathscr B\epsilon^{-1/2}(-iY)\psi(s) &= \int_\mathbb R\mathrm dx\ \mathscr B(s, x)(-1)\epsilon^{-1/2}iY\psi(x)\\
	&= \int_\mathbb R\mathrm dx\ \mathscr B(s, x)(-1)\epsilon^{1/2}\partial_x\psi(x)\\
	&= \int_\mathbb R\mathrm dx\ \epsilon^{1/2}\partial_x\mathscr B(s, x)\psi(x)\\
	&= (2^{-1/2}\epsilon^{-1/2}s -2^{-1/2}\epsilon^{1/2}\partial_s)\mathscr B\psi(s).
	\end{align}
	\begin{align}
	\epsilon^{-1/2}X &= 2^{-1/2}\mathcal S^\dag +2^{-1/2}\mathcal S,\\
	\epsilon^{-1/2}(-iY) &= 2^{-1/2}\mathcal S^\dag -2^{-1/2}\mathcal S,
	\end{align}
	\begin{align}
	&\epsilon^{-1/2}(X -iY) = 2^{1/2}\mathcal S^\dag,\quad \mathcal S^\dag = (2\epsilon)^{-1/2}(X -iY),\\
	&\epsilon^{-1/2}(X +iY) = 2^{1/2}\mathcal S,\quad \mathcal S = (2\epsilon)^{-1/2}(X +iY).
	\end{align}
	We then conclude
	\begin{align}
	&X = 2^{-1/2}\epsilon^{1/2}(\mathcal S^\dag +\mathcal S),\quad Y = 2^{-1/2}\epsilon^{1/2}i(\mathcal S^\dag -\mathcal S),\\
	&\mathcal S^\dag = (2\epsilon)^{-1/2}(X -iY),\quad \mathcal S = (2\epsilon)^{-1/2}(X +iY),\\
	&\mathscr B\mathcal S^\dag = \epsilon^{-1/2}s\mathscr B,\quad \mathscr B\mathcal S = \epsilon^{1/2}\partial_s\mathscr B.
	\end{align}
	\end{proof}
	
	\begin{rem}
	One should be careful to note that the $s$ variable found in the BS representation given here is not related to $\zeta = x +iy$ but instead to $(2\epsilon)^{-1/2}\overline \zeta$ since
	\begin{align}
	&\mathscr B\mathscr W\overline \zeta = \mathscr B\mathscr W(x -iy) = \mathscr B(X -iY) = \mathscr B\mathcal S^\dag = (2\epsilon)^{1/2}s,\\
	&\mathscr B\mathscr W\zeta = \mathscr B\mathscr W(x +iy) = \mathscr B(X +iY) = \mathscr B\mathcal S = (2\epsilon)^{1/2}\partial_s,
	\end{align}
	where $\zeta, \overline \zeta \in \mathscr C$ are considered as tempered distributions on phase space. Furthermore
	\begin{align}
	&\mathscr Wm_{\overline \zeta} = \mathscr Wm_{x -iy} = M_{X -iY}\mathscr W = (2\epsilon)^{1/2}M_{\mathcal S^\dag}\mathscr W,\\
	&\mathscr Wm_\zeta = \mathscr Wm_{x +iy} = M_{X +iY}\mathscr W = (2\epsilon)^{1/2}M_\mathcal S\mathscr W.
	\end{align}
	\end{rem}
	The occurrence of inconveniences in notation like that noted above is likely inevitable as it can appear that attempting to eliminate them in one place cause other inconveniences to emerge in another place.
	
	\begin{defn}
	We denote by $L^{(\alpha)}_n : \mathbb R_+ \to \mathbb R$, $n \in \mathbb Z_+$, $-1 < \alpha \in \mathbb R$, $L^{(\alpha)}_n : x \mapsto L^{(\alpha)}_n(x)$, the \emph{generalized Laguerre polynomials} \cite[\S~18.5]{DLMF}. They satisfy
	\begin{align}
	&L^{(\alpha)}_n(x) = \sum_{k = 0}^n\binom{n +\alpha}{k +\alpha}\frac{(-x)^k}{k!},\\
	&L^{(\alpha)}_n(0) = \binom{n +\alpha}n = \Gamma(n +1)^{-1}\Gamma(\alpha +1)^{-1}\Gamma(n +\alpha +1).
	\end{align}
	Define $\mathcal C^{(\alpha)}$ to be the discrete operator formally defined by
	\begin{align}
	&\mathcal C^{(\alpha)}v_n =-(n +1)v_{n +1}(x) +(2n +\alpha +1)v_n(x) -(n + \alpha)v_{n -1}
	\end{align}
	for all $-1 < \alpha \in \mathbb R$, $v: \mathbb Z_+ \to \mathbb C$, $v: n \mapsto v_n$.
	\end{defn}
	
	\begin{rem}
	The generalized Laguerre polynomials, $L^{(\alpha)}_n$, satisfy \cite[\S~18.5]{DLMF}:
	\begin{align}
	&\int_0^\infty\mathrm dx\ x^\alpha e^{-x}L^{(\alpha)}_m(x)L^{(\alpha)}_n(x) = \frac{\Gamma(n +\alpha +1)}{n!}\delta_{m, n},\\
	&\mathcal C^{(\alpha)}L^{(\alpha)}_n(x) = xL^{(\alpha)}_n(x).
	\end{align}
	The operator $\mathcal C^{(\alpha)}$ is singular at $n = 0$ for $\alpha = 0$ alone and therefore additional boundary conditions must be specified at $n = 0$ for $\alpha \ne 0$.
	\end{rem}
	
	\begin{defn}
	We denote by $l^{(\alpha)}_n : \mathbb R_+ \to \mathbb R$, $n \in \mathbb Z_+$, $-1 < \alpha \in \mathbb R$, $l^{(\alpha)}_n : x \mapsto L^{(\alpha)}_n(x)$, the \emph{generalized Laguerre functions} define by:
	\begin{align}
	&l^{(\alpha)}_n(x) := \Gamma(n +\alpha +1)^{-1/2}\Gamma(n +1)^{1/2}e^{-x/2}L^{(\alpha)}_n(x).
	\end{align}
	Define $\mathcal D^{(\alpha)}$ to be discrete operator formally defined by
	\begin{align}
	\mathcal D^{(\alpha)}v_n &= -(n +\alpha +1)^{1/2}(n +1)^{1/2}v_{n +1}(x) +(2n +\alpha +1)v_n(x)\\
	&\qquad -(n +\alpha)^{1/2}n^{1/2}v_{n -1}(x),
	\end{align}
	for all $-1 < \alpha \in \mathbb R$, $v: \mathbb Z_+ \to \mathbb C$, $v: n \mapsto v_n$.
	\end{defn}
	
	\begin{rem}
	\begin{align}
	&\int_0^\infty\mathrm dx\ l^{(\alpha)}_m(x)l^{(\alpha)}_n(x) = \delta_{m, n},\\
	&\mathcal D^{(\alpha)}l^{(\alpha)}_n(x) = xl^{(\alpha)}_n(x).
	\end{align}
	The operator $\mathcal D^{(\alpha)}$ is singular at $n = 0$ for all $-1 < \alpha \in \mathbb R$ so boundary conditions needn't be specified. Furthermore $\mathcal D^{(\alpha)}$ is manifestly symmetric with respect to the usual lattice inner product. The $l^{(\alpha)}_n$ also satisfy
	\begin{align}
	\sum_{n = 0}^\infty l^{(\alpha)}_n(x)l^{(\alpha)}_n(y) &= \delta(x -y)
	\end{align}
	in the distributional sense in $L^2(\mathbb R_+, \mathbb C)$ for all $\alpha$.
	\end{rem}
	
	\begin{defn}
	We define $\{l_{m, n}\}_{m, n = 0}^\infty$ to be an orthonormal basis of $\mathscr C$ specified by
	\begin{align}
	&l_{m, n}(x, y) \equiv l_{m, n}^\mathrm c(\zeta) := (2^{-1/2}i\zeta)^{m -n}l^{(m -n)}_n(|\zeta|^2/2),\quad m \ge n,\\
	&l_{m, n}(x, y) \equiv l_{m, n}^\mathrm c(\zeta) := (2^{-1/2}i\overline \zeta)^{n -m}l^{(n -m)}_m(|\zeta|^2/2),\quad m < n,\\
	&\qquad \zeta = x +iy.
	\end{align}
	\end{defn}
	
	\begin{prop}
	It is the case that
	\begin{align}
	&\langle h_m, W_{\xi_x, \xi_y}h_n\rangle = l_{m, n}(\epsilon^{1/2}\xi_x, \epsilon^{1/2}\xi_y).
	\end{align}
	\end{prop}
	We follow the method of proof in \cite{Fo89} and carry through explicitly since in our notation the results differ.
	
	\begin{proof}
	We recall that
	\begin{align}
	\langle\phi, W_{\xi_x, \xi_y}\psi\rangle = \int_\mathbb R\mathrm dx\ e^{i\xi_xx}\overline\phi(x -\epsilon\xi_y/2)\psi(x +\epsilon\xi_y/2),
	\end{align}
	for all Schwartz class $\phi, \psi \in \mathscr H$. One may then find
	\begin{align}
	&\langle h_m, W_{\xi_x, \xi_y}h_n\rangle = \int_\mathbb R\mathrm dx\ e^{i\xi_xx}h_m(x -\epsilon\xi_y/2)h_n(x +\epsilon\xi_y/2)\\
	&= \int_\mathbb R\mathrm dx\ e^{i\xi_x(x +\epsilon\xi_y/2)}h_m(x)h_n(x +\epsilon\xi_y)\\
	&= \int_\mathbb R\mathrm dx\ e^{i\xi_x(x +\epsilon\xi_y/2)}h_m(x)h_n(x +\epsilon\xi_y)
	\end{align}
	\begin{align}
	&= e^{i\epsilon\xi_x\xi_y/2}[\epsilon^{-1/4}\pi^{-1/4}2^{-m/2}(m!)^{-1/2}][\epsilon^{-1/4}\pi^{-1/4}2^{-n/2}(n!)^{-1/2}]\\
	&\qquad \times\int_\mathbb R\mathrm dx\ e^{i\xi_xx}e^{-x^2/2\epsilon}H_m(\epsilon^{-1/2}x)e^{-(x +\epsilon\xi_y)^2/2\epsilon}H_n(\epsilon^{-1/2}\{x +\epsilon\xi_y\})\\
	&= \epsilon^{-1/2}e^{i\epsilon\xi_x\xi_y/2}[\pi^{-1/4}2^{-m/2}(m!)^{-1/2}][\pi^{-1/4}2^{-n/2}(n!)^{-1/2}]\\
	&\qquad \times\int_\mathbb R\mathrm dx\ e^{i\xi_xx}e^{-x^2/2\epsilon}e^{-x^2/2\epsilon -x\xi_y -\epsilon\xi_y^2/2}H_m(\epsilon^{-1/2}x)H_n(\epsilon^{-1/2}x +\epsilon^{1/2}\xi_y)\\
	&= \epsilon^{-1/2}e^{i\epsilon\xi_x\xi_y/2 -\epsilon\xi_y^2/2}[\pi^{-1/4}2^{-m/2}(m!)^{-1/2}][\pi^{-1/4}2^{-n/2}(n!)^{-1/2}]\\
	&\qquad \times\int_\mathbb R\mathrm dx\ e^{-x^2/\epsilon -\xi_yx +i\xi_xx}H_m(\epsilon^{-1/2}x)H_n(\epsilon^{-1/2}x +\epsilon^{1/2}\xi_y)\\
	&= e^{i\epsilon\xi_x\xi_y/2 -\epsilon\xi_y^2/2}[\pi^{-1/4}2^{-m/2}(m!)^{-1/2}][\pi^{-1/4}2^{-n/2}(n!)^{-1/2}]\\
	&\qquad \times\int_\mathbb R\mathrm dx\ e^{-x^2 -\epsilon^{1/2}\xi_yx +i\epsilon^{1/2}\xi_xx}H_m(x)H_n(x +\epsilon^{1/2}\xi_y).
	\end{align}
	One may observe that
	\begin{align}
	&\int_\mathbb R\mathrm dx\ e^{-x^2 -\epsilon^{1/2}\xi_yx +i\epsilon^{1/2}\xi_xx}H_m(x)H_n(x +\epsilon^{1/2}\xi_y)\\
	&\qquad = \int_\mathbb R\mathrm dx\ e^{-x^2 +\epsilon^{1/2}(i\xi_x -\xi_y)x -2^{-2}\epsilon(i\xi_x -\xi_y)^2 +2^{-2}\epsilon(i\xi_x -\xi_y)^2}H_m(x)H_n(x +\epsilon^{1/2}\xi_y)\\
	&\qquad = e^{2^{-2}\epsilon(i\xi_x -\xi_y)^2}\int_\mathbb R\mathrm dx\ e^{-[x -2^{-1}\epsilon^{1/2}(i\xi_x -\xi_y)]^2}H_m(x)H_n(x +\epsilon^{1/2}\xi_y)\\
	&\qquad = e^{2^{-2}\epsilon(i\xi_x -\xi_y)^2}\int_{\mathbb R +2^{-1}\epsilon^{1/2}(i\xi_x -\xi_y)}\mathrm dy\ e^{-y^2}H_m(y +2^{-1}\epsilon^{1/2}\{i\xi_x -\xi_y\})\\
	&\qquad\qquad \times H_n(y +2^{-1}\epsilon^{1/2}\{i\xi_x +\xi_y\}),\quad y = x -2^{-1}\epsilon^{1/2}(i\xi_x -\xi_y),\\
	&\qquad = e^{2^{-2}\epsilon(i\xi_x -\xi_y)^2}\int_\mathbb R\mathrm dy\ e^{-y^2}H_m(y +2^{-1}\epsilon^{1/2}\{i\xi_x -\xi_y\})\\
	&\qquad\qquad \times H_n(y +2^{-1}\epsilon^{1/2}\{i\xi_x +\xi_y\}),
	\end{align}
	where we have deformed the path $\mathbb R +2^{-1}\epsilon^{1/2}(i\xi_x -\xi_y)$ to the path $\mathbb R$ with impunity due to the absence of singular points in the interior of the complex plane and the vanishing of the integrand at the endpoints of the path for all fixed $\xi_x, \xi_y \in \mathbb R$. We proceed
	\begin{align}
	&\qquad = e^{-2^{-2}\epsilon(\xi_x +i\xi_y)^2}\int_\mathbb R\mathrm dy\ e^{-y^2}H_m(y +2^{-1}i\epsilon^{1/2}\{\xi_x +i\xi_y\})H_n(y +2^{-1}i\epsilon\{\xi_x -i\xi_y\})\\
	&\qquad = e^{-2^{-2}\epsilon\xi_\zeta^2}\int_\mathbb R\mathrm dy\ e^{-y^2}H_m(y +2^{-1}\epsilon^{1/2}i\xi_\zeta)H_n(y +2^{-1}\epsilon^{1/2}i\overline\xi_\zeta),\quad \xi_\zeta = \xi_x +i\xi_y,
	\end{align}
	\begin{align}
	&= e^{-2^{-2}\epsilon\xi_\zeta^2}\int_\mathbb R\mathrm dx\ e^{-x^2}\sum_{j = 0}^m\binom mj(\epsilon^{1/2}i\xi_\zeta)^{m -j}H_j(x)\sum_{k = 0}^n\binom nk(\epsilon^{1/2}i\overline\xi_\zeta)^{n -k}H_k(x)\\
	&= e^{-2^{-2}\epsilon\xi_\zeta^2}\sum_{j = 0}^m\binom mj(\epsilon^{1/2}i\xi_\zeta)^{m -j}\sum_{k = 0}^n\binom nk(\epsilon^{1/2}i\overline\xi_\zeta)^{n -k}\int_\mathbb R\mathrm dx\ e^{-x^2}H_j(x)H_k(x)\\
	&= e^{-2^{-2}\epsilon\xi_\zeta^2}\sum_{j = 0}^m\binom mj(\epsilon^{1/2}i\xi_\zeta)^{m -j}\sum_{k = 0}^n\binom nk(\epsilon^{1/2}i\overline\xi_\zeta)^{n -k}\pi^{1/2}2^jj!\delta_{j, k}.
	\end{align}
	where we have used that
	\begin{align}
	H_n(x +y) = \sum_{j = 0}^n\binom nk(2y)^{n -k}H_j(x).
	\end{align}

	The above sum is 0 if $j > n$ or $k > m$. This in turn means that one should replace both $m, n$ with $\min(m, n)$. However, since the factor $\delta_{j, k}$ can be found above, only one sum needs to carry this as an upper bound. One may therefore conclude that
	\begin{align}
	&e^{-2^{-2}\xi_\zeta^2}\sum_{j = 0}^m\binom mj(\epsilon^{1/2}i\xi_\zeta)^{m -j}\sum_{k = 0}^n\binom nk(\epsilon^{1/2}i\overline\xi_\zeta)^{n -k}\pi^{1/2}2^jj!\delta_{j, k}\\
	&\qquad= e^{-2^{-2}\epsilon\xi_\zeta^2}\sum_{j = 0}^{\min\{m, n\}}\binom mj(\epsilon^{1/2}i\xi_\zeta)^{m -j}\binom nj(\epsilon^{1/2}i\overline\xi_\zeta)^{n -k}\pi^{1/2}2^jj!.
	\end{align}

	First take $m \ge n$ and $j = n -k$.
	\begin{align}
	&e^{-2^{-2}\epsilon\xi_\zeta^2}\sum_{j = 0}^{\min\{m, n\}}\binom mj(\epsilon^{1/2}i\xi_\zeta)^{m -j}\binom nj(\epsilon^{1/2}i\overline\xi_\zeta)^{n -j}\pi^{1/2}2^jj!\\
	&\qquad = e^{-2^{-2}\epsilon\xi_\zeta^2}\sum_{j = 0}^n\binom mj(\epsilon^{1/2}i\xi_\zeta)^{m -j}\binom nj(\epsilon^{1/2}i\overline\xi_\zeta)^{n -j}\pi^{1/2}2^jj!\\
	&\qquad = e^{-2^{-2}\epsilon\xi_\zeta^2}\sum_{k = 0}^n\binom m{n - k}(\epsilon^{1/2}i\xi_\zeta)^{m -n +k}\binom n{n - k}(\epsilon^{1/2}i\overline\xi_\zeta)^k\pi^{1/2}2^{n -k}(n -k)!\\
	&\qquad = \pi^{1/2}(\epsilon^{1/2}i\xi_\zeta)^{m -n}2^ne^{-2^{-2}\epsilon\xi_\zeta^2}\sum_{k = 0}^n\binom m{n -k}(\epsilon^{1/2}i\xi_\zeta)^k\\
	&\qquad\qquad \times\binom n{n -k}(\epsilon^{1/2}i\overline\xi_\zeta)^k2^{-k}(n -k)!
	\end{align}
	\begin{align}
	&\qquad = \pi^{1/2}(\epsilon^{1/2}i\xi_\zeta)^{m -n}2^ne^{-2^{-2}\epsilon\xi_\zeta^2}\sum_{k = 0}^n\binom m{n -k}\binom n{n -k}(n -k)!(-2)^{-k}\epsilon^k|\xi_\zeta|^{2k}\\
	&\qquad = \pi^{1/2}(\epsilon^{1/2}i\xi_\zeta)^{m -n}2^ne^{-2^{-2}\epsilon\xi_\zeta^2}\sum_{k = 0}^n(-1)^k\frac{m!n!}{(n -k)!(m -n +k)!k!}(2^{-1}\epsilon|\xi_\zeta|^2)^k\\
	&\qquad = \pi^{1/2}n!(\epsilon^{1/2}i\xi_\zeta)^{m -n}2^ne^{-2^{-2}\epsilon\xi_\zeta^2}\sum_{k = 0}^n(-1)^k\frac{[n +(m -n)]!}{(n -k)!(m -n +k)!k!}(2^{-1}\epsilon|\xi_\zeta|^2)^k\\
	&\qquad = \pi^{1/2}n!(\epsilon^{1/2}i\xi_\zeta)^{m -n}2^ne^{-2^{-2}\epsilon\xi_\zeta^2}L_n^{(m -n)}(2^{-1}\epsilon|\xi_\zeta|^2),
	\end{align}
	where we have used that:
	\begin{align}
	L_n^{(\alpha)}(x) &= \sum_{j = 0}^n(-1)^j\binom{n +\alpha}{n -j}\frac{x^j}{j!} = \sum_{j = 0}^n(-1)^j\frac{(n +\alpha)!}{(n -j)!(\alpha +j)!j!}x^j.
	\end{align}
	One may then find
	\begin{align}
	&\langle h_m, W_{\xi_x, \xi_y}h_n\rangle\\
	&\qquad = e^{i\epsilon\xi_x\xi_y/2 -\epsilon\xi_y^2/2}[\pi^{-1/4}2^{-m/2}(m!)^{-1/2}][\pi^{-1/4}2^{-n/2}(n!)^{-1/2}]\\
	&\qquad\qquad \times\int_\mathbb R\mathrm dx\ e^{-x^2 -\epsilon^{1/2}\xi_yx +i\epsilon^{1/2}\xi_xx}H_m(x)H_n(x +\epsilon^{1/2}\xi_y)\\
	&\qquad = e^{i\epsilon\xi_x\xi_y/2 -\epsilon\xi_y^2/2}[\pi^{-1/4}2^{-m/2}(m!)^{-1/2}][\pi^{-1/4}2^{-n/2}(n!)^{-1/2}]\\
	&\qquad\qquad \times\pi^{1/2}n!(\epsilon^{1/2}i\xi_\zeta)^{m -n}2^ne^{-2^{-2}\epsilon\xi_\zeta^2}L_n^{(m -n)}(2^{-1}\epsilon|\xi_\zeta|^2)\\
	&\qquad = e^{i\epsilon\xi_x\xi_y/2 -\epsilon\xi_y^2/2}2^{-(m -n)/2}(m!)^{-1/2}(n!)^{1/2}\\
	&\qquad\qquad \times(i\xi_\zeta)^{m -n}e^{2^{-2}(\epsilon\xi_y^2 -2i\epsilon\xi_y\xi_x -\epsilon\xi_x^2)}L_n^{(m -n)}(2^{-1}\epsilon|\xi_\zeta|^2)
	\end{align}
	\begin{align}
	&\qquad = 2^{-(m -n)/2}(m!)^{-1/2}(n!)^{1/2}(\epsilon^{1/2}i\xi_\zeta)^{m -n}\\
	&\qquad\qquad \times e^{2^{-1}i\epsilon\xi_x\xi_y -2^{-1}\epsilon\xi_y^2 +2^{-2}\epsilon\xi_y^2 -2^{-1}i\epsilon\xi_y\xi_x -2^{-2}\xi_x^2}L_n^{(m -n)}(2^{-1}\epsilon|\xi|^2)\\
	&\qquad = 2^{-(m -n)/2}(m!)^{-1/2}(n!)^{1/2}(\epsilon^{1/2}i\xi_\zeta)^{m -n}e^{-2^{-2}\epsilon\xi_y^2 -2^{-2}\epsilon\xi_x^2}L_n^{(m -n)}(2^{-1}\epsilon|\xi|^2)\\
	&\qquad = 2^{-(m -n)/2}(\epsilon^{1/2}i\xi_\zeta)^{m -n}(m!)^{-1/2}(n!)^{1/2}e^{-\epsilon|\xi_\zeta|^2/2^2}L_n^{(m -n)}(\epsilon|\xi_\zeta|^2/2).
	\end{align}
	Then since
	\begin{align}
	l^{(\alpha)}_n(x) &= \Gamma(n +\alpha +1)^{-1/2}\Gamma(n +1)^{1/2}e^{-x/2}L^{(\alpha)}_n(x)\\
	&= [(n +\alpha)!]^{-1/2}(n!)^{1/2}e^{-x/2}L^{(\alpha)}_n(x),
	\end{align}
	it is the case that
	\begin{align}
	\langle h_m, W_{\xi_x, \xi_y}h_n\rangle &= (2^{-1/2}\epsilon^{1/2}i\xi_\zeta)^{m -n}l_n^{(m -n)}(\epsilon|\xi_\zeta|^2/2)\\
	&= l_{m, n}^\mathrm c(\epsilon^{1/2}\xi_\zeta) = l_{m, n}(\epsilon^{1/2}\xi_x, \epsilon^{1/2}\xi_y).
	\end{align}

	Now take $m \le n$ and $j = m -k$. We observe that doing so is equivalent to simultaneously taking $(m, n) \mapsto (n, m)$ and $\xi_\zeta \mapsto \overline\xi_\zeta$ inside the sum. We may then immediately write
	\begin{align}
	&\langle h_m, W_{\xi_x, \xi_y}h_n\rangle = (2^{-1/2}\epsilon^{1/2}i\overline\xi_\zeta)^{n -m}l_n^{(n -m)}(\epsilon|\xi_\zeta|^2/2)\\
	&= l_{m, n}^\mathrm c(\epsilon^{1/2}\xi_\zeta) = l_{m, n}(\epsilon^{1/2}\xi_x, \epsilon^{1/2}\xi_y),
	\end{align}
	which which coincides with the expression for the case of $m \ge n$ for $m = n$ and therefore completes the proof.
	\end{proof}
		
	It is useful to consider the commutative basis whose image under the Weyl transform is $\{E_{m, n} \}_{m, n = 0}^\infty$.
	\begin{prop}
	It is the case that
	\begin{align}
	e_{m, n}(x, y) &= 2\epsilon^{-1}(-1)^ml_{n, m}(2\epsilon^{-1/2}y, -2\epsilon^{-1/2}x).
	\end{align}
	\end{prop}
	
	\begin{defn}
	We define
	\begin{align}
	e_{m, n} := \mathscr W^{-1}E_{m, n}.
	\end{align}
	\end{defn}
	
	\begin{proof}
	One may find
	\begin{align}
	e_{m, n}(x, y) &= \mathscr W^{-1}E_{m, n}(x, y)\\
	&= \mathscr W^{-1}(h_m\otimes h_n^*)(x, y)\\
	&= 2\epsilon^{-1}\langle h_n, W_{2\epsilon^{-1}y, -2\epsilon^{-1}x}\widehat h_m\rangle(x, y)\\
	&= 2\epsilon^{-1}(-1)^m\langle h_n, W_{2\epsilon^{-1}y, -2\epsilon^{-1}x}h_m\rangle(x, y)\\
	&= 2\epsilon^{-1}(-1)^ml_{n, m}(2\epsilon^{-1/2}y, -2\epsilon^{-1/2}x),
	\end{align}
	where we have used that $H_n(-x) = (-1)^nH_n(x)$.
	\end{proof}

	\section{The Laplacian and the quadratic potential}
	
	\begin{defn}
	Let
	\begin{align}
	l &:= -\Delta = \epsilon^{-2}p_x^2 +\epsilon^{-2}p_y^2 = -\partial_x^2 -\partial_y^2,\\
	L &:= \mathscr Wl\mathscr W^{-1} = \epsilon^{-2}P_X^2 +\epsilon^{-2}P_Y^2 = \epsilon^{-2}\mathrm{ad}_Y^2 +\epsilon^{-2}\mathrm{ad}_X^2.
	\end{align}
	\end{defn}
	
	\begin{prop}
	It is the case that
	\begin{align}
	&2^{-1}\epsilon LE_{m, n} = -(m +1)^{1/2}(n +1)^{1/2}E_{m +1, n +1} +(m +n +1)E_{m, n}\\
	&\qquad -m^{1/2}n^{1/2}E_{m -1, n -1}.
	\end{align}
	\end{prop}
	
	\begin{proof}
	The Jacobi identity gives
	\begin{align}
	&0 = [A, [B, C] ] +[B ,[C, A] ] +[C, [A, B] ]\\
	&[A, [B, C] ] -[B, [A, C] ] = [ [A, B], C]
	\end{align}
	and therefore
	\begin{align}
	[\mathrm{ad}_A, \mathrm{ad}_B]C &= \mathrm{ad}_A\mathrm{ad}_BC -\mathrm{ad}_B\mathrm{ad}_AC,\\
	&= [A, [B, C] ] -[B, [A, C] ]\\
	&= [ [A, B], C]\\
	&= \mathrm{ad}_{[A, B]}C.
	\end{align}
	Then since
	\begin{align}
	&\mathcal S = (2\epsilon)^{-1/2}(X +iY),\quad \mathcal S^\dag = (2\epsilon)^{-1/2}(X -iY),
	\end{align}
	one has that
	\begin{align}
	[\mathcal S, \mathcal S^\dag] &= [(2\epsilon)^{-1/2}(X +iY), (2\epsilon)^{-1/2}(X -iY)],\\
	&= (2\epsilon)^{-1}([X, -iY] +[iY, X])\\
	&= -i\epsilon^{-1}[X, Y] = I\\
	[\mathrm{ad}_\mathcal S, \mathrm{ad}_{\mathcal S^\dag}] &= \mathrm{ad}_{[\mathcal S, \mathcal S^\dag]} = \mathrm{ad}_I = 0.
	\end{align}
	Since
	\begin{align}
	&X = 2^{-1/2}\epsilon^{1/2}(\mathcal S^\dag +\mathcal S),\quad Y = 2^{-1/2}\epsilon^{1/2}i(\mathcal S^\dag -\mathcal S).
	\end{align}
	one has
	\begin{align}
	\mathrm{ad}_X^2 &= 2^{-1}\epsilon\mathrm{ad}_{\mathcal S^\dag +\mathcal S}^2 = 2^{-1}\epsilon(\mathrm{ad}_{\mathcal S^\dag} +\mathrm{ad}_\mathcal S)(\mathrm{ad}_{\mathcal S^\dag} +\mathrm{ad}_\mathcal S)\\
	&= 2^{-1}\epsilon(\mathrm{ad}_{\mathcal S^\dag}^2 +2\epsilon\mathrm{ad}_{\mathcal S^\dag}\mathrm{ad}_\mathcal S +\mathrm{ad}_\mathcal S^2),\\
	\mathrm{ad}_Y^2 &= -2^{-1}\epsilon\mathrm{ad}_{\mathcal S^\dag -\mathcal S}^2 = -2^{-1}\epsilon(\mathrm{ad}_{\mathcal S^\dag} -\mathrm{ad}_\mathcal S)(\mathrm{ad}_{\mathcal S^\dag} -\mathrm{ad}_\mathcal S)\\
	&= -2^{-1}\epsilon(\mathrm{ad}_{\mathcal S^\dag}^2 -2\mathrm{ad}_{\mathcal S^\dag}\mathrm{ad}_\mathcal S +\mathrm{ad}_\mathcal S^2).
	\end{align}
	\begin{align}
	\epsilon^2L &= \mathrm{ad}_X^2 +\mathrm{ad}_Y^2 = 2\epsilon\mathrm{ad}_\mathcal S\mathrm{ad}_{\mathcal S^\dag} = 2\epsilon\mathrm{ad}_{\mathcal S^\dag}\mathrm{ad}_\mathcal S.
	\end{align}
	One can then straightforwardly compute the action of $L$ on the Hermite operators.
	\begin{align}
	&2^{-1}\epsilon LE_{m, n} = \mathrm{ad}_\mathcal S\mathrm{ad}_{\mathcal S^\dag}E_{m, n} = \mathrm{ad}_\mathcal S[\mathcal S^\dag, h_m\otimes h_n^*]\\
	&\qquad = \mathrm{ad}_\mathcal S[(\mathcal S^\dag h_m)\otimes h_n^* -h_m\otimes(\mathcal S h_n)^*]\\
	&\qquad = \mathrm{ad}_\mathcal S[(m +1)^{1/2}h_{m +1}\otimes h_n^* -n^{1/2}h_m\otimes h_{n -1}^*]\\
	&\qquad = (m +1)^{1/2}(\mathcal S h_{m +1})\otimes h_n^* -(m +1)^{1/2}h_{m +1}\otimes(\mathcal S^\dag h_n)^*\\
	&\qquad\qquad -n^{1/2}(\mathcal S h_m)\otimes h_{n - 1}^* +n^{1/2}h_m\otimes(\mathcal S^\dag h_{n -1})^*\\
	&\qquad = (m +1)h_m\otimes h_n^* -(m +1)^{1/2}(n +1)^{1/2}h_{m +1}\otimes h_{n +1}^*\\
	&\qquad\qquad -m^{1/2}n^{1/2}h_{m -1}\otimes h_{n -1}^* +nh_m\otimes h_n^*\\
	&\qquad = (m +n +1)h_m\otimes h_n^* -(m +1)^{1/2}(n +1)^{1/2}h_{m +1}\otimes h_{n +1}^*\\
	&\qquad\qquad -m^{1/2}n^{1/2}h_{m -1}\otimes h_{n -1}^*\\
	&\qquad = -(m +1)^{1/2}(n +1)^{1/2}E_{m +1, n +1} +(m +n +1)E_{m, n} -m^{1/2}n^{1/2}E_{m -1, n -1}.
	\end{align}
	\end{proof}

	When working with linear operators on $\mathscr C$ alone it is common to conflate $m_f$ with $f$ for a general potential function $f \in \mathscr C$. This is no longer harmless when working with the Weyl transform. Although $\mathscr Wm_x\mathscr W^{-1} = M_X, \mathscr Wm_y\mathscr W^{-1} = M_Y$, for a generic $f \in \mathscr C$ it is the case that $\mathscr Wm_f\mathscr W^{-1} \ne M_F$ in general. We will demonstrate this with a study of the quadratic potential, which in the context on deformation quantization is the harmonic oscillator Hamiltonian.
	
	\begin{defn}
	Let
	\begin{align}
	r^2 &:= x^2 +y^2,\quad R^2 := \mathscr Wr^2 = X^2 +Y^2,\\
	v &:= m_{r^2} = m_x^2 +m_y^2,\quad V := \mathscr Wv\mathscr W^{-1} = M_X^2 +M_Y^2.
	\end{align}
	\end{defn}
	
	\begin{prop}
	One has that
	\begin{align}
	R^2 &= \epsilon(2\mathcal S^\dag\mathcal S +I),\quad R^2h_n = \epsilon(2n +1)h_n.
	\end{align}
	\end{prop}
	
	\begin{cor}
	It is the case that
	\begin{align}
	M^\mathrm L_{R^2}E_{m, n} &= \epsilon(2m +1)E_{m, n},\quad M^\mathrm R_{R^2}E_{m, n} = \epsilon(2n +1)E_{m, n}\\
	M_{R^2}E_{m, n} &= \epsilon(m +n +1)E_{m, n}.
	\end{align}
	\end{cor}
	
	\begin{proof}
	We recall that
	\begin{align}
	&X = 2^{-1/2}\epsilon^{1/2}(\mathcal S^\dag +\mathcal S),\quad Y = 2^{-1/2}\epsilon^{1/2}i(\mathcal S^\dag -\mathcal S).
	\end{align}
	Therefore
	\begin{align}
	2\epsilon^{-1}X^2 &= (\mathcal S^\dag +\mathcal S)^2 = \mathcal S^{\dag, 2} +\mathcal S^\dag \mathcal S +\mathcal S\mathcal S^\dag +\mathcal S^2,\\
	2\epsilon^{-1}Y^2 &= -(\mathcal S^\dag -\mathcal S)^2 = -\mathcal S^{\dag, 2} +\mathcal S^\dag \mathcal S +\mathcal S\mathcal S^\dag -\mathcal S^2,\\
	\epsilon^{-1}(X^2 +Y^2) &= \mathcal S^\dag \mathcal S +\mathcal S\mathcal S^\dag = 2\mathcal S^\dag \mathcal S -[\mathcal S^\dag, \mathcal S] = 2\mathcal S^\dag \mathcal S +[\mathcal S, \mathcal S^\dag]\\
	&= 2\mathcal S^\dag\mathcal S +I.
	\end{align}
	Then
	\begin{align}
	R^2h_n &= \epsilon(2\mathcal S^\dag\mathcal S +I)h_n = 2\epsilon\mathcal S^\dag n^{1/2}h_{n -1} +\epsilon h_n = \epsilon(2n +1)h_n.
	\end{align}
	\end{proof}
	
	\begin{prop}
	One has
	\begin{align}
	4\epsilon^{-1}V &= M^\mathrm L_{\mathcal S^\dag}M^\mathrm R_\mathcal S +2\epsilon^{-1}M_{R^2} +M^\mathrm L_\mathcal S M^\mathrm R_{\mathcal S^\dag},\\
	4\epsilon^{-1}VE_{m, n} &= (m +1)^{1/2}(n +1)^{1/2}E_{m +1, n +1} +2(m +n +1)E_{m, n} +m^{1/2}n^{1/2}E_{m -1, n -1}.
	\end{align}
	\end{prop}
	
	\begin{proof}
	We recall that
	\begin{align}
	&X = 2^{-1/2}\epsilon^{1/2}(\mathcal S^\dag +\mathcal S),\quad Y = 2^{-1/2}\epsilon^{1/2}i(\mathcal S^\dag -\mathcal S).
	\end{align}
	and therefore
	\begin{align}
	8\epsilon^{-1}M_X^2A &= 2\epsilon^{-1}(X^2A +2XAX +AX^2)\\
	&= (\mathcal S^\dag +\mathcal S)^2A +2(\mathcal S^\dag +\mathcal S)A(\mathcal S^\dag +\mathcal S) +A(\mathcal S^\dag +\mathcal S)^2\\
	&= (\mathcal S^{\dag 2} +\mathcal S^\dag \mathcal S +\mathcal S\mathcal S^\dag +\mathcal S^2)A +2(\mathcal S^\dag A\mathcal S^\dag +\mathcal S^\dag A\mathcal S +\mathcal S A\mathcal S^\dag +\mathcal S A\mathcal S)\\
	&\qquad +A(\mathcal S^{\dag, 2} +\mathcal S^\dag \mathcal S +\mathcal S\mathcal S^\dag +\mathcal S^2),
	\end{align}
	\begin{align}
	8\epsilon^{-1}M_Y^2A &= -2\epsilon^{-1}(Y^2A +2YAY +AY^2)\\
	&= -(\mathcal S^\dag -\mathcal S)^2A -2(\mathcal S^\dag -\mathcal S)A(\mathcal S^\dag -\mathcal S) -A(\mathcal S^\dag -\mathcal S)^2\\
	&= -(\mathcal S^{\dag, 2} -\mathcal S^\dag \mathcal S -\mathcal S\mathcal S^\dag +\mathcal S^2)A -2(\mathcal S^\dag A\mathcal S^\dag -\mathcal S^\dag A\mathcal S -\mathcal S A\mathcal S^\dag +\mathcal S A\mathcal S)\\
	&\qquad -A(\mathcal S^{\dag, 2} -\mathcal S^\dag\mathcal S -\mathcal S\mathcal S^\dag +\mathcal S^2)
	\end{align}
	and
	\begin{align}
	4\epsilon^{-1}(M_X^2 +M_Y^2)A &= (\mathcal S^\dag\mathcal S +\mathcal S\mathcal S^\dag)A +(\mathcal S^\dag A\mathcal S +\mathcal S A\mathcal S^\dag) +A(\mathcal S^\dag\mathcal S +\mathcal S\mathcal S^\dag)\\
	&= (2\mathcal S^\dag\mathcal S -[\mathcal S^\dag, \mathcal S])A +A(2\mathcal S^\dag\mathcal S -[\mathcal S^\dag, \mathcal S]) +(\mathcal S^\dag A\mathcal S +\mathcal S A\mathcal S^\dag)\\
	&= (2\mathcal S^\dag\mathcal S +I)A +A(2\mathcal S^\dag\mathcal S +I) +(\mathcal S^\dag A\mathcal S +\mathcal S A\mathcal S^\dag)\\
	&= (2\epsilon^{-1}M_{R^2} +M^\mathrm L_{\mathcal S^\dag}M^\mathrm R_\mathcal S +M^\mathrm L_\mathcal S M^\mathrm R_{\mathcal S^\dag})A.
	\end{align}
	We then find
	\begin{align}
	&4\epsilon^{-1}(M_X^2 +M_Y^2)E_{m, n}\\
	&\qquad = (2\epsilon^{-1}M_{R^2} +M^\mathrm L_{\mathcal S^\dag}M^\mathrm R_\mathcal S +M^\mathrm L_\mathcal S M^\mathrm R_{\mathcal S^\dag})h_m\otimes h_n^*\\
	&\qquad = \epsilon^{-1}(R^2h_m)\otimes h_n^* +\epsilon^{-1}h_m\otimes(R^2h_n)^* +(\mathcal S^\dag h_m)\otimes(\mathcal S^\dag h_n)^* +(\mathcal S h_m)\otimes(\mathcal S h_n )^*\\
	&\qquad = (2m +1)h_m\otimes h_n^* +(2n +1)h_m\otimes h_n^* +(m +1)^{1/2}(n +1)^{1/2}h_{m +1}\otimes h_{n +1}^*\\
	&\qquad\qquad +m^{1/2}n^{1/2}h_{m -1}\otimes h_{n -1}^*\\
	&\qquad = 2(m +n +1)E_{m, n} +(m +1)^{1/2}(n +1)^{1/2}E_{m +1, n +1} +m^{1/2}n^{1/2}E_{m -1, n -1}.
	\end{align}
	\end{proof}

	The actions of $2^{-1}\epsilon^2L$ and $4\epsilon^{-1}V$ on the Hermite operators $E_{m, n}$ are very similar and their sum produces an automorphism that acts diagonally on this operator basis.
	
	\begin{defn}
	We define
	\begin{align}
	h &:= 2^{-1}\epsilon l +4\epsilon^{-1}v,\quad H := \mathscr Wh\mathscr W^{-1} = 2^{-1}\epsilon L +4\epsilon^{-1}V.
	\end{align}
	\end{defn}
	In the physics literature, the operator $h$ would be referred to as a 2D harmonic oscillator Hamiltonian.
	\begin{rem}
	One has that
	\begin{align}
	HE_{m, n} &= 3(m +n +1)E_{m, n}.
	\end{align}
	\end{rem}

	\section{Polar expressions for distinguished bases}
	
	\begin{defn}
	For $j \in \mathbb Z$, $n \in \mathbb Z_+$, we define
	\begin{align}
	\mathcal E_{j, n} &= E_{n +j, n},\quad j \ge 0,\qquad \mathcal E_{j, n} = E_{n, n -j},\quad j < 0,\\
	\varepsilon_{j, n} &= e_{n +j, n},\quad j \ge 0,\qquad \varepsilon_{j, n} = e_{n, n -j},\quad j < 0,\\
	\ell_{|j|, n}(\varrho) &= (2^{-1/2}i\varrho)^{|j|}l^{(|j|)}_n(\varrho^2/2),\quad \varrho \in \mathbb R_+,
	\end{align}
	where $\xi_\zeta = \xi_x +i\xi_y = \varrho e^{i\vartheta}$.
	\end{defn}
	The significance of $j = m -n$, considered in the context of $E_{m, n}$, as a index for angular momentum was addressed thoroughly in \cite{Ac13}, albeit through rather different methodology.
	\begin{prop}
	It is the case that
	\begin{align}
	W_{\xi_x, \xi_y} = \sum_{j \in \mathbb Z}\sum_{n = 0}^\infty\ell_{j, n}(\epsilon^{1/2}\varrho)e^{ij\vartheta}\mathcal E_{j, n}.
	\end{align}
	\end{prop}
	
	\begin{proof}
	We recall that
	\begin{align}
	W_{\xi_x, \xi_y} = \sum_{m, n = 0}^\infty l_{m, n}(\epsilon^{1/2}\xi_x, \epsilon^{1/2}\xi_y)E_{m, n},
	\end{align}
	where
	\begin{align}
	&l_{m, n}(\xi_x, \xi_y) = (2^{-1/2}i\xi_\zeta)^{m -n}l^{(m -n)}_n(|\xi_\zeta|^2/2),\quad m \ge n,\\
	&l_{m,n}(\xi_x, \xi_y) = (2^{-1/2}i\overline\xi_\zeta)^{n -m}l^{(n -m)}_m(|\xi_\zeta|^2/2),\quad m < n,\\
	&\qquad \xi_\zeta = \xi_x +i\xi_y,
	\end{align}
	and in polar coordinates
	\begin{align}
	&l_{m, n}'(\varrho, \vartheta) = (2^{-1/2}i\varrho e^{i\vartheta})^{m -n}l^{(m -n)}_n(\varrho^2/2),\quad m \ge n,\\
	&l_{m,n}'(\varrho, \vartheta) = (2^{-1/2}i\varrho e^{-i\vartheta})^{n -m}l^{(n -m)}_m(\varrho^2/2),\quad m < n,\\
	&\qquad \xi_\zeta = \xi_x +i\epsilon\xi_y = \varrho e^{i\vartheta}.
	\end{align}
	
	We then find for $j \ge 0$:
	\begin{align}
	l_{n +j, n}'(\varrho, \vartheta) &= (2^{-1/2}i\varrho e^{i\vartheta})^jl^{(j)}_n(\varrho^2/2)\\
	&= (2^{-1/2}i\varrho)^jl^{(j)}_n(\varrho^2/2)e^{ij\vartheta}\\
	&= (2^{-1/2}i\varrho)^{|j|}l^{(|j|)}_n(\varrho^2/2)e^{ij\vartheta}\\
	&= \ell_{|j|, n}(\varrho)e^{ij\vartheta},
	\end{align}
	and for $j < 0$:
	\begin{align}
	l_{n, n -j}'(\varrho, \vartheta) &= (2^{-1/2}i\varrho e^{-i\vartheta})^{-j}l^{(-j)}_n(\varrho^2/2)\\
	&= (2^{-1/2}i\varrho)^{-j}l^{(-j)}_n(\varrho^2/2)e^{ij\vartheta}\\
	&= (2^{-1/2}i\varrho)^{|j|}l^{(|j|)}_n(\varrho^2/2)e^{ij\vartheta}\\
	&= \ell_{|j|, n}(\varrho)e^{ij\vartheta}.
	\end{align}
	
	\begin{align}
	W_{\xi_x, \xi_y} &= \sum_{m, n = 0}^\infty l_{m, n}(\epsilon^{1/2}\xi_x, \epsilon^{1/2}\xi_y)E_{m, n},\\
	W_{\varrho, \vartheta}' &= \sum_{m, n = 0}^\infty l_{m, n}'(\epsilon^{1/2}\varrho, \vartheta)E_{m, n}\\
	&= \sum_{j = 0}^\infty\sum_{n = 0}^\infty E_{n +j, n}l_{n +j, n}'(\epsilon^{1/2}\varrho, \vartheta) +\sum_{j = -\infty}^{-1}\sum_{n = 0}^\infty E_{n, n -j}l_{n, n -j}'(\epsilon^{1/2}\varrho, \vartheta)\\
	&= \sum_{j = 0}^\infty\sum_{n = 0}^\infty\mathcal E_{j, n}\ell_{|j|, n}(\epsilon^{1/2}\varrho)e^{ij\vartheta} +\sum_{j = -\infty}^{-1}\sum_{n = 0}^\infty\mathcal E_{j, n}\ell_{|j|, n}(\epsilon^{1/2}\varrho)e^{ij\vartheta}
	\end{align}
	\begin{align}
	&= \sum_{j \in \mathbb Z}\sum_{n = 0}^\infty\ell_{|j|, n}(\epsilon^{1/2}\varrho)e^{ij\vartheta}\mathcal E_{j, n}.
	\end{align}
	\end{proof}
	
	\begin{prop}
	One has that
	\begin{align}
	\varepsilon_{j, n}'(\rho, \theta) &= 2\epsilon^{-1}(-1)^ni^{-|j|}\ell_{|j|, n}(2\epsilon^{-1/2}\rho)e^{-ij\theta},
	\end{align}
	for all $j, n, \rho, \theta$, where $\zeta = x +i\epsilon^{-1}y = \rho e^{i\theta}$.
	\end{prop}
	
	\begin{proof}
	We observe that
	\begin{align}
	\zeta &= x +iy\\
	\zeta' &= y -ix = -i(x +iy) = -i\zeta,\\
	\overline \zeta' &= i\overline \zeta.
	\end{align}
	One then has
	\begin{align}
	e_{m, n}(x, y) &= 2\epsilon^{-1}(-1)^ml_{n, m}(2\epsilon^{-1/2}y, -2\epsilon^{-1/2}x)\\
	&= 2\epsilon^{-1}(-1)^m[2^{-1/2}i(2\epsilon^{-1/2}\zeta')]^{n -m}l^{(n -m)}_m(|2\epsilon^{-1/2}\zeta'|^2/2),\quad n \ge m,\\
	&= 2\epsilon^{-1}(-1)^m[2^{-1/2}i(-2\epsilon^{-1/2}i\zeta)]^{n -m}l^{(n -m)}_m(|2\epsilon^{-1/2}\zeta|^2/2)\\
	&= 2\epsilon^{-1}(-1)^mi^{m -n}[2^{-1/2}i(2\epsilon^{-1/2}\zeta)]^{n -m}l^{(n -m)}_m(|2\epsilon^{-3/2}\zeta|^2/2)\\
	&= 2\epsilon^{-1}i^{-(m +n)}[2^{-1/2}i(2\epsilon^{-1/2}\zeta)]^{n -m}l^{(n -m)}_m(|2\epsilon^{-1/2}\zeta|^2/2),
	\end{align}
	
	\begin{align}
	e_{m, n}(x, y) &= 2\epsilon^{-1}(-1)^ml_{n, m}(2\epsilon^{-1/2}y, -2\epsilon^{-1/2}x)\\
	&= 2\epsilon^{-1}(-1)^m[2^{-1/2}i(2\epsilon^{-1/2}\overline \zeta')]^{m -n}l^{(m -n)}_n(|2\epsilon^{-1/2}\overline \zeta'|^2/2),\quad n < m,\\
	&= 2\epsilon^{-1}(-1)^m[2^{-1/2}i(2\epsilon^{-1/2}i\overline \zeta)]^{m -n}l^{(m -n)}_n(|2\epsilon^{-1/2}\zeta|^2/2)\\
	&= 2\epsilon^{-1}(-1)^mi^{m -n}[2^{-1/2}i(2\epsilon^{-1/2}\overline \zeta)]^{m -n}l^{(m -n)}_n(|2\epsilon^{-1/2}\zeta|^2/2)\\
	&= 2\epsilon^{-1}i^{-(m +n)}[2^{-1/2}i(2\epsilon^{-1/2}\overline \zeta)]^{m -n}l^{(m -n)}_n(|2\epsilon^{-1/2}\zeta|^2/2),
	\end{align}
	and in polar coordinates
	\begin{align}
	e_{m, n}'(\rho, \theta) &= 2\epsilon^{-1}i^{-(m +n)}[2^{-1/2}i(2\epsilon^{-1/2}\rho e^{i\theta})]^{n -m}l^{(n -m)}_m(\{2\epsilon^{-1/2}\rho\}^2/2),\quad n \ge m,\\
	e_{m, n}'(\rho, \theta) &= 2\epsilon^{-1}i^{-(m +n)}[2^{-1/2}i(2\epsilon^{-1/2}\rho e^{-i\theta})]^{m -n}l^{(m -n)}_n(\{2\epsilon^{-1/2}\rho\}^2/2),\quad n < m.
	\end{align}
	Furthermore
	\begin{align}
	e_{m, m -j}'(\rho, \theta) &= 2\epsilon^{-1}i^{-(2m -j)}[2^{-1/2}i(2\epsilon^{-1/2}\rho e^{i\theta})]^{-j}l^{(-j)}_m(\{2\epsilon^{-1/2}\rho\}^2/2),\quad j \le 0,\\
	&= 2\epsilon^{-1}i^{-(2m +|j|)}[2^{-1/2}i(2\epsilon^{-1/2}\rho)]^{|j|}l^{(|j|)}_m(\{2\epsilon^{-1/2}\rho\}^2/2)e^{-ij\theta}\\
	&= 2\epsilon^{-1}(-1)^mi^{-|j|}\ell_{|j|, m}(2\epsilon^{-1/2}\rho)e^{-ij\theta},
	\end{align}
	\begin{align}
	e_{n +j, n}'(\rho, \theta) &= 2\epsilon^{-1}i^{-(2n +j)}[2^{-1/2}i(2\epsilon^{-1/2}\rho e^{-i\theta})]^jl^{(j)}_n(\{2\epsilon^{-1/2}\rho\}^2/2),\quad j > 0,\\
	&= 2\epsilon^{-1}i^{-(2n +|j|)}[2^{-1/2}i(2\epsilon^{-1/2}\rho)]^{|j|}l^{(|j|)}_n(\{2\epsilon^{-1/2}\rho\}^2/2)e^{-ij\theta}\\
	&= 2\epsilon^{-1}(-1)^ni^{-|j|}\ell_{|j|, n}(2\epsilon^{-1/2}\rho)e^{-ij\theta},
	\end{align}
	and therefore
	\begin{align}
	\varepsilon_{j, n}'(\rho, \theta) &= 2\epsilon^{-1}(-1)^ni^{-|j|}\ell_{|j|, n}(2\epsilon^{-1/2}\rho)e^{-ij\theta}.
	\end{align}
	\end{proof}

	\section{Spectral vectors of the Laplacian}
	
	\begin{defn}
	Let $\omega_{\lambda, j} : \mathbb R^2 \to \mathbb C$, where $\lambda \in \mathbb R_+$, $j \in \mathbb Z$, and $\omega_{\lambda, j}' : (\rho, \theta) \mapsto \omega_{\lambda, j}'(\rho, \theta)$, be given by
	\begin{align}
	\omega_{\lambda, j}'(\rho, \theta) := i^{-|j|}J_{|j|}(\lambda^{1/2}\rho)e^{ij\theta},
	\end{align}
	where $\zeta = x +iy = \rho e^{i\theta}$. Furthermore let
	\begin{align}
	\Omega_{\lambda, j} := (-1)^j\sum_{n = 0}^\infty\ell_{|j|, n}(\epsilon^{1/2}\lambda^{1/2})\mathcal E_{-j, n}.
	\end{align}
	\end{defn}
	
	\begin{prop}
	The spectral vectors, $\phi_\lambda$, of the Laplacian, $l$, may be represented by
	\begin{align}
	\phi_\lambda &= (2\pi)^{-1/2}\sum_{j \in \mathbb Z}\phi_{\lambda, j}\omega_{\lambda, j},\quad \sum_{j \in \mathbb Z}|\phi_{\lambda, j}|^2 = 1,
	\end{align}
	and satisfy
	\begin{align}
	&l\phi_\lambda = \lambda\phi_\lambda,\quad \int_{\mathbb R^2}\rho\mathrm d\rho\mathrm d\theta\ \phi_\lambda(\rho, \theta)\overline \phi_\eta(\rho, \theta) = \delta(\lambda -\eta),
	\end{align}
	in the distributional sense.
	\end{prop}
	
	\begin{proof}
	We note that the Bessel functions of the first kind, $J_\alpha(\zeta)$, satisfy \cite[p.~361]{AS}:
	\begin{align}
	e^{(t -1/t)\zeta/2} &= \sum_{n \in \mathbb Z}t^nJ_n(\zeta),\quad \zeta \in \mathbb C, t \in \mathbb C\setminus\{0\},
	\end{align}
	and therefore for $t = -ie^{i\theta}$ one has
	\begin{align}
	e^{(-ie^{i\theta} -ie^{-i\theta})x/2} &= e^{-ix\cos\theta} = \sum_{j \in \mathbb Z}(-i)^je^{ij\theta}J_j(x),\quad x \in \mathbb R, \theta \in [0, 2pi).
	\end{align}
	The Fourier transform can then be written in polar form as:
	\begin{align}
	\mathscr Ff(x, y) &= (2\pi)^{-1}\int_{\mathbb R^2}\mathrm d\xi_x\mathrm d\xi_y\ e^{-i(x\xi_x +y\xi_y)}f(\xi_x, \xi_y)\\
	\mathscr Ff'(\rho, \theta) &= (2\pi)^{-1}\int_{\mathbb R^2}\varrho\mathrm d\varrho\mathrm d\vartheta\ e^{-i\varrho \rho\cos(\theta -\vartheta)}f'(\varrho, \vartheta)\\
	&= (2\pi)^{-1}\int_{\mathbb R^2}\varrho\mathrm d\varrho\mathrm d\vartheta\ \sum_{j \in \mathbb Z}(-i)^je^{ij(\theta -\vartheta)}J_j(\varrho \rho)f'(\varrho, \vartheta).
	\end{align}
	
	As is well-known, the spectral vectors of the Laplacian, $l$, may be written as
	\begin{align}
	l\phi_\lambda &= \lambda\phi_\lambda,\quad \varrho^2\varphi_\lambda = \lambda\varphi_\lambda,\quad \lambda \in \mathbb R_+,\quad \phi_\lambda = \mathscr F\varphi_\lambda,
	\end{align}
	and therefore
	\begin{align}
	\varphi'(\varrho, \vartheta) &= 2\delta(\varrho^2 -\lambda)\varphi_\lambda'(\vartheta)\\
	&= \lambda^{-1/2}[\delta(\varrho -\lambda^{1/2}) +\delta(\varrho +\lambda^{1/2})]\varphi_\lambda'(\vartheta),
	\end{align}
	for some choice of $\varphi_\lambda$. We also note that \cite[\S~10.4]{DLMF}
	\begin{align}
	J_{-n}(\zeta) &= (-1)^nJ_n(\zeta),\quad \zeta \in \mathbb C,
	\end{align}
	then since
	\begin{align}
	i^{-(-j)}J_{-j}(x) = i^j(-1)^jJ_j(x) = i^{-j}J_j(x),
	\end{align}
	it is the case that
	\begin{align}
	i^{-j}J_j(x) = i^{-|j|}J_{|j|}(x).
	\end{align}
	One may then observe
	\begin{align}
	\phi_\lambda'(\rho, \theta) &= \mathscr F\varphi'(\rho, \theta)\\
	&= (2\pi)^{-1}\int_0^\infty\varrho\mathrm d\varrho\int_0^{2\pi}\mathrm d\vartheta\ \sum_{j \in \mathbb Z}(-i)^je^{ij(\theta -\vartheta)}J_j(\varrho\rho)\varphi'(\varrho, \vartheta)\\
	&= (2\pi)^{-1}\int_0^\infty\varrho\mathrm d\varrho\int_0^{2\pi}\mathrm d\vartheta\ \sum_{j \in \mathbb Z}(-i)^je^{ij(\theta -\vartheta)}J_j(\varrho\rho)\\
	&\qquad \times\lambda^{-1/2}[\delta(\varrho -\lambda^{1/2}) +\delta(\varrho +\lambda^{1/2})]\varphi_\lambda'(\vartheta)
	\end{align}
	\begin{align}
	&= (2\pi)^{-1}\int_0^{2\pi}\mathrm d\vartheta\ \sum_{j \in \mathbb Z}(-i)^je^{ij(\theta -\vartheta)}J_j(\lambda^{1/2}\rho)\varphi_\lambda'(\vartheta)\\
	&= (2\pi)^{-1}\sum_{j \in \mathbb Z}J_j(\lambda^{1/2}\rho)(-i)^je^{ij\theta}\int_0^{2\pi}\mathrm d\vartheta\ e^{-ij\vartheta}\varphi_\lambda'(\vartheta)\\
	&= (2\pi)^{-1/2}\sum_{j \in \mathbb Z}\phi_{\lambda, j}i^{-|j|}J_{|j|}(\lambda^{1/2}\rho)e^{ij\theta},
	\end{align}
	where each step should be taken in the distributional sense and therefore can be justified weakly on a sufficiently well behaved, densely defined, subspace of $\mathscr C$, and where we have denoted
	\begin{align}
	\phi_{\lambda, j} = (2\pi)^{-1/2}\int_0^{2\pi}\mathrm d\vartheta\ e^{-ij\vartheta}\varphi_\lambda'(\vartheta).
	\end{align}
	
	We recall that the $J_\alpha(\zeta)$ satisfy \cite[\S~1.17]{DLMF}:
	\begin{align}
	\delta(x -y) &= \int_0^\infty\mathrm dt\ xtJ_\alpha(xt)J_\alpha(yt)\\
	&= \int_0^\infty2\mathrm du\ xJ_\alpha(xu^{1/2})J_\alpha(yu^{1/2}),\quad u = t^2,\quad \mathrm du = 2t\mathrm dt,\\
	(2x)^{-1}\delta(x -y) &= \int_0^\infty\mathrm du\ J_\alpha(xu^{1/2})J_\alpha(yu^{1/2}),
	\end{align}
	and therefore
	\begin{align}
	\delta(x -y) &= \int_0^\infty\mathrm dt\ xtJ_\alpha(xt)J_\alpha(yt)\\
	&= \int_0^\infty\mathrm dt\ (xy)^{1/2}tJ_\alpha(xt)J_\alpha(yt)\\
	&= \int_0^\infty\mathrm du\ uJ_\alpha(x\{xy\}^{-1/4}u)J_\alpha(y\{xy\}^{-1/4}u),\quad t = (xy)^{-1/4}u,\\
	&= \int_0^\infty\mathrm du\ uJ_\alpha(x^{1/2}u)J_\alpha(y^{1/2}u)\\
	&= \int_0^\infty\mathrm dv\ J_\alpha(\{2xv\}^{1/2})J_\alpha(\{2yv\}^{1/2}),\quad v = 2^{-1}u^2,\quad \mathrm dv = u\mathrm du.
	\end{align}
	Lastly, one would like to normalize the spectral vectors in the distributional sense. To this end one may find
	\begin{align}
	&\int_0^\infty\rho\mathrm d\rho\int_0^{2\pi}\mathrm dy\ \phi_\lambda(\rho, \theta)\overline\phi_\eta(\rho, \theta)\\
	&\qquad = (2\pi)^{-1}\int_0^\infty\mathrm d\rho\ \rho\sum_{j, k \in \mathbb Z}\phi_{\lambda, j}J_j(\lambda^{1/2}\rho)\overline\phi_{\eta, k}J_k(\eta^{1/2}\rho)\int_0^{2\pi}\mathrm d\theta\ e^{i(j -k)\theta}\\
	&\qquad = \int_0^\infty\mathrm d\rho\ \rho\sum_{j, k \in \mathbb Z}\phi_{\lambda, j}\overline\phi_{\eta, k}J_j(\lambda^{1/2}\rho)J_k(\eta^{1/2}\rho)\delta_{j, k}
	\end{align}
	\begin{align}
	&\qquad = \sum_{j \in \mathbb Z}\phi_{\lambda, j}\overline\phi_{\eta, j}\int_0^\infty\mathrm d\rho\ \rho J_j(\lambda^{1/2}\rho)J_j(\eta^{1/2}\rho)\\
	&\qquad = \sum_{j \in \mathbb Z}\phi_{\lambda, j}\overline\phi_{\eta, j}\delta(\lambda -\eta)\\
	&\qquad = \delta(\lambda -\eta)\sum_{j \in \mathbb Z}|\phi_{\lambda, j}|^2,
	\end{align}
	where, again, each step must be taken in the appropriate distributional sense. We observe that
	\begin{align}
	\sum_{j \in \mathbb Z}|\phi_{\lambda, j}|^2 &= \int_0^{2\pi}\mathrm d\vartheta\ |\varphi_\lambda'(\vartheta)|^2,
	\end{align}
	require
	\begin{align}
	\sum_{j \in \mathbb Z}|\phi_{\lambda, j}|^2 = 1,
	\end{align}
	and find
	\begin{align}
	\int_{\mathbb R^2}\rho\mathrm d\rho\mathrm d\theta\ \phi_\lambda'(\rho, \theta)\overline \phi_\eta'(\rho, \theta) = \delta(\lambda -\eta).
	\end{align}
	\end{proof}
	
	An alternative, and more standard, approach involves using separation of variables to reduce the original spectral equation to the Bessel equation via a suitable ansatz.
	\begin{lem}
	For $\zeta = x +iy = \rho e^{i\theta}$, t is the case that
	\begin{align}
	\partial_\zeta &= 2^{-1}\partial_x -2^{-1}i\partial_y = 2^{-1}e^{-i\theta}\partial_\rho -i2^{-1}\rho^{-1}e^{-i\theta}\partial_\theta,\\
	\partial_{\overline \zeta} &= 2^{-1}\partial_x +2^{-1}i\partial_y = 2^{-1}e^{i\theta}\partial_\rho +i2^{-1}\rho^{-1}e^{i\theta}\partial_\theta.
	\end{align}
	\end{lem}
	
	\begin{proof}
	\begin{align}
	\zeta &= \rho e^{i\theta},\quad \rho^2 = \zeta\overline \zeta,\\
	e^{i\theta} &= \frac{\zeta}{|\zeta|} = \frac{\zeta^{1/2}}{\overline \zeta^{1/2}},\quad i\theta = 2^{-1}\log \zeta -2^{-1}\log\overline \zeta,\\
	x &= 2^{-1}\zeta +2^{-1}\overline \zeta,\quad y = -2^{-1}i\zeta +2^{-1}i\overline \zeta,
	\end{align}
	
	\begin{align}
	\partial_\zeta &= \partial_\zeta x\partial_x +\partial_\zeta y\partial_y\\
	&= \partial_\zeta(2^{-1}\zeta +2^{-1}\overline \zeta)\partial_x +\partial_\zeta(-2^{-1}i\zeta +2^{-1}i\overline \zeta)\partial_y\\
	&= 2^{-1}\partial_x -2^{-1}i\partial_y,\\
	\partial_{\overline \zeta} &= 2^{-1}\partial_x +2^{-1}i\partial_y,
	\end{align}
	
	\begin{align}
	\partial_{\rho^2}\rho &= (\partial_\rho\rho^2)^{-1} = 2^{-1}\rho^{-1},\\
	\partial_{e^{i\theta} }\theta &= (\partial_\theta e^{i\theta})^{-1} = -ie^{-i\theta},
	\end{align}
	
	\begin{align}
	\partial_{\rho^2} &= \partial_{\rho^2}\zeta\partial_\zeta +\partial_{\rho^2}\overline \zeta\partial_{\overline \zeta}\\
	\partial_{\rho^2}\rho\partial_\rho&= 2^{-1}\rho^{-1}e^{i\theta}\partial_\zeta +2^{-1}\rho^{-1}e^{-i\theta}\partial_{\overline \zeta}\\
	2^{-1}\rho^{-1}\partial_\rho &= 2^{-1}\rho^{-1}e^{i\theta}\partial_\zeta +2^{-1}\rho^{-1}e^{-i\theta}\partial_{\overline \zeta}\\
	\rho\partial_\rho &= \rho e^{i\theta}\partial_\zeta +\rho e^{-i\theta}\partial_{\overline \zeta}\\
	\rho\partial_\rho &= \zeta\partial_\zeta +\overline \zeta\partial_{\overline \zeta},
	\end{align}
	
	\begin{align}
	\partial_{e^{i\theta} } &= \partial_{e^{i\theta} }\zeta\partial_\zeta +\partial_{e^{i\theta} }\overline \zeta\partial_{\overline \zeta}\\
	\partial_{e^{i\theta} }\theta\partial_\theta &= \partial_{e^{i\theta} }\zeta\partial_\zeta +\partial_{e^{i\theta} }\overline \zeta\partial_{\overline \zeta}\\
	-ie^{-i\theta}\partial_\theta &= \rho\partial_\zeta -\rho e^{-2i\theta}\partial_{\overline \zeta}\\
	-i\partial_\theta &= \rho e^{i\theta}\partial_\zeta -\rho e^{-i\theta}\partial_{\overline \zeta}\\
	&= \zeta\partial_\zeta -\overline \zeta\partial_{\overline \zeta},
	\end{align}
	
	\begin{align}
	2\zeta\partial_\zeta &= \rho\partial_\rho -i\partial_\theta\\
	\partial_\zeta &= 2^{-1}e^{-i\theta}\partial_\rho -i2^{-1}\rho^{-1}e^{-i\theta}\partial_\theta\\
	\partial_{\overline \zeta} &= 2^{-1}e^{i\theta}\partial_\rho +i2^{-1}\rho^{-1}e^{i\theta}\partial_\theta.
	\end{align}
	\end{proof}
	
	\begin{prop}
	The radial part of the equation $l\phi_\lambda = \lambda\phi_\lambda$ takes the form of the Bessel equation via the method of separation of variables and allows one to arrive at the same solution as above.
	\end{prop}
	
	\begin{proof}
	\begin{align}
	-\partial_\zeta\partial_{\overline \zeta}f = -(2^{-1}\partial_x -i2^{-1}\partial_y)(2^{-1}\partial_x +i2^{-1}\partial_y)f = 2^{-2}(-\partial_x^2 -\partial_y^2)f = 2^{-2}lf.
	\end{align}
	
	\begin{align}
	lf &= -2^2\partial_\zeta\partial_{\overline \zeta}f\\
	&= -2^2(2^{-1}e^{-i\theta}\partial_\rho -i2^{-1}\rho^{-1}e^{-i\theta}\partial_\theta)(2^{-1}e^{i\theta}\partial_\rho +i2^{-1}\rho^{-1}e^{i\theta}\partial_\theta)f\\
	&= -2^2(2^{-2}\partial_\rho^2 -i2^{-2}\rho^{-2}\partial_\theta +i2^{-2}\rho^{-1}\partial_\rho\partial_\theta +2^{-2}\rho^{-1}\partial_\rho\\
	&\qquad -i2^{-2}\rho^{-1}\partial_\rho\partial_\theta +i2^{-2}\rho^{-2}\partial_\theta +2^{-2}\rho^{-2}\partial_\theta^2)f\\
	&= -2^2(2^{-2}\partial_\rho^2 +2^{-2}\rho^{-1}\partial_\rho +2^{-2}\rho^{-2}\partial_\theta^2)f\\
	&= (-\partial_\rho^2 -\rho^{-1}\partial_\rho -\rho^{-2}\partial_\theta^2)f,
	\end{align}
	
	\begin{align}
	l &= -\partial_\rho^2 -\rho^{-1}\partial_\rho -\rho^{-2}\partial_\theta^2.
	\end{align}
	Take $f_j'(\rho, \theta) = g_j'(\rho)e^{ij\theta}$ then
	\begin{align}
	lf_j &= (-\partial_\rho^2 -\rho^{-1}\partial_\rho -\rho^{-2}\partial_\theta^2)f_j\\
	&= (-\partial_\rho^2 -\rho^{-1}\partial_\rho +j^2\rho^{-2})f_j,\\
	lg_j &= (-\partial_\rho^2 -\rho^{-1}\partial_\rho +j^2\rho^{-2})g_j.
	\end{align}
	Now consider
	\begin{align}
	\lambda g_j &= lg_j,\\
	&= (-\partial_\rho^2 -\rho^{-1}\partial_\rho +j^2\rho^{-2})g_j\\
	0 &= (\partial_\rho^2 +\rho^{-1}\partial_\rho +\lambda -j^2\rho^{-2})g_j\\
	0 &= (\rho^2\partial_\rho^2 +\rho\partial_\rho +\lambda\rho^2 -j^2)g_j\\
	&= [(\lambda^{1/2}\rho)^2\partial_{\lambda^{1/2}\rho}^2 +(\lambda^{1/2}\rho)\partial_{\lambda^{1/2}\rho} +(\lambda^{1/2}\rho)^2 -j^2]g_j\\
	&= (x^2\partial_x^2 +x\partial_x +x^2 -j^2)g_j,\quad x = \lambda^{1/2}\rho,
	\end{align}
	which is Bessel's equation and has a complete set of solutions given by Bessel functions of the first kind. Therefore $g_j'(\rho) = J_j(x) = J_j(\lambda^{1/2}\rho)$. Due to linearity one may, up to overall scale, represent a more general solution, $\phi_\lambda$, as
	\begin{align}
	\phi'_\lambda(\rho, \theta) &= (2\pi)^{-1/2}\sum_{j \in \mathbb Z}\phi_{\lambda, j}\omega_{\lambda, j}'(\rho, \theta),\quad \sum_{j \in \mathbb Z}|\phi_{\lambda, j}|^2 = 1.
	\end{align}
	\end{proof}
	
	\begin{prop}
	Let $\phi_\lambda$ be the spectral vectors of the Laplacian, $l$, represented by
	\begin{align}
	\phi_\lambda &= (2\pi)^{-1/2}\sum_{j \in \mathbb Z}\phi_{\lambda, j}\omega_{\lambda, j},\quad \sum_{j \in \mathbb Z}|\phi_{\lambda, j}|^2 = 1.
	\end{align}
	The Weyl transform of the spectral vectors, $\Phi_\lambda = \mathscr W\phi_\lambda$, have the representation
	\begin{align}
	\Phi_\lambda &= (2\pi)^{-1/2}\sum_{j \in \mathbb Z}\phi_{\lambda, j}\Omega_{\lambda, j}.
	\end{align}
	\end{prop}
	
	\begin{cor}
	One has that $\mathscr W\omega_{\lambda, j} = \Omega_{\lambda, j}$, for all $\lambda, j$.
	\end{cor}
	
	\begin{proof}
	We observe that since the Weyl transform can act as a unitary bijection, that
	\begin{align}
	l\phi_\lambda = \lambda\phi_\lambda\quad \iff\quad L\Phi_\lambda = \lambda\Phi_\lambda
	\end{align}
	and proceed along these lines. We recall that
	\begin{align}
	2^{-1}\epsilon LE_{m, n} &= -(m +1)^{1/2}(n +1)^{1/2}E_{m +1, n +1} +(m +n +1)E_{m, n}\\
	&\qquad -m^{1/2}n^{1/2}E_{m -1, n -1}
	\end{align}
	and therefore for $j = m -n \ge 0$:
	\begin{align}
	2^{-1}\epsilon L\mathcal E_{j, n} &= -(m +1)^{1/2}(n +1)^{1/2}\mathcal E_{m -n, n +1} +(m +n +1)\mathcal E_{m -n, n}\\
	&\qquad -m^{1/2}n^{1/2}\mathcal E_{m -n, n -1}\\
	&= -(n +1)^{1/2}(n +j +1)^{1/2}\mathcal E_{j, n +1} +(2n +j +1)\mathcal E_{j, n}\\
	&\qquad -n^{1/2}(n +j)^{1/2}\mathcal E_{j, n -1}
	\end{align}
	and for $j = m -n \le 0$:
	\begin{align}
	2^{-1}\epsilon L\mathcal E_{j, n} &= -(n +1)^{1/2}(n -j +1)^{1/2}\mathcal E_{j, n +1} +(2n -j +1)\mathcal E_{j, n}\\
	&\qquad -n^{1/2}(n -j)^{1/2}\mathcal E_{j, n -1}.
	\end{align}
	These equations may be combined into a single one which holds for all $j = m -n \in \mathbb Z$:
	\begin{align}
	2^{-1}\epsilon L\mathcal E_{j, n} &= - (n +1)^{1/2}(n +|j| +1)^{1/2}\mathcal E_{j, n +1} + (2n +|j| +1)\mathcal E_{j, n}\\
	&\qquad -n^{1/2}(n |j|)^{1/2}\mathcal E_{j, n -1}.
	\end{align}
	Since $L$ is manifestly symmetric, one may reformulate the above equation on basis operators into an equation on the coefficients of an element $\Phi_\lambda = \sum_{j \in \mathbb Z}\sum_{n = 0}^\infty f_{\lambda, j, n}\mathcal E_{j, n}$, $\Phi_{\lambda, j, n} \in \mathbb C$, in the form of the spectral equation
	\begin{align}
	2^{-1}\epsilon Lf_{\lambda, j, n} &= -(n +1)^{1/2}(n +|j| +1)^{1/2}f_{\lambda, j, n +1} +(2n +|j| +1)f_{\lambda, j, n}\\
	&\qquad -n^{1/2}(n +|j|)^{1/2}f_{\lambda, j, n -1}\\
	&= 2^{-1}\epsilon\lambda f_{\lambda, j, n}
	\end{align}
	and select a scaling of $f_{\lambda, j}$ so that
	\begin{align}
	Lf_{\lambda, j} &= (\epsilon\lambda)f_{\lambda, j}.
	\end{align}
	We have from the recurrence formula of the Laguerre functions that the solution can be represented as
	\begin{align}
	f_{\lambda, -j, n} &= (2\pi)^{-1/2}(-1)^j\Phi_{\lambda, j}\ell_{|j|, n}(\epsilon^{1/2}\lambda^{1/2}),\quad \Phi_{\lambda, j, n} \in \mathbb C,\\
	\Phi_\lambda &= \sum_{j \in \mathbb Z}\sum_{n = 0}^\infty f_{\lambda, -j, n}\mathcal E_{j, n} = \sum_{j \in \mathbb Z}\sum_{n = 0}^\infty f_{\lambda, j, n}\mathcal E_{-j, n},\\
	&= (2\pi)^{-1/2}\sum_{j \in \mathbb Z}(-1)^j\Phi_{\lambda, j}\sum_{n = 0}^\infty \ell_{|j|, n}(\epsilon^{1/2}\lambda^{1/2})\mathcal E_{-j, n},
	\end{align}
	for some suitable choice of $\Phi_{\lambda, j, n}$.
	
	We require that
	\begin{align}
	\Phi_\lambda &= \mathscr W\phi_\lambda.
	\end{align}
	We recall
	\begin{align}
	l\phi_\lambda &= \lambda\phi_\lambda,\quad \varrho^2\varphi_\lambda = \lambda\varphi_\lambda,\quad \lambda \in \mathbb R_+,\quad \phi_\lambda = \mathscr F\varphi_\lambda,
	\end{align}
	\begin{align}
	\varphi'(\varrho, \vartheta) &= 2\delta(\varrho^2 -\lambda)\varphi_\lambda'(\vartheta)\\
	&= \lambda^{-1/2}[\delta(\varrho -\lambda^{1/2}) +\delta(\varrho +\lambda^{1/2})]\varphi_\lambda'(\vartheta),
	\end{align}
	for some suitable choice of $\varphi_\lambda$ and where $\xi_\zeta = \xi_x +i\xi_y = \varrho e^{i\vartheta}$. We further recall that
	\begin{align}
	W_{\xi_x, \xi_y} &= \sum_{m, n = 0}^\infty l_{m, n}(\epsilon^{1/2}\xi_x, \epsilon^{1/2}\xi_y)E_{m, n},\\
	W_{\varrho, \vartheta}' &= \sum_{j \in \mathbb Z}\sum_{n = 0}^\infty\ell_{|j|, n}(\epsilon^{1/2}\varrho)e^{ij\vartheta}\mathcal E_{j, n},
	\end{align}
	where
	\begin{align}
	\ell_{|j|, n}(\varrho) &:= (2^{-1/2}i\varrho)^{|j|}l^{(|j|)}_n(\rho^2/2),\quad \xi_x +i\xi_y = \varrho e^{i\vartheta}
	\end{align}
	and remark that
	\begin{align}
	\mathscr F^2f(x, y) &= f(-x, -y),\quad \mathscr F^2f'(\rho, \theta) = f'(\rho, \theta +\pi).
	\end{align}
	Then
	\begin{align}
	\mathscr Wf &= (2\pi)^{-1}\int_{\mathbb R^2}\mathrm d\xi_x\mathrm d\xi_y\ W_{\xi_x, \xi_y}\mathscr Ff(\xi_x, \xi_y)\\
	&= (2\pi)^{-1}\int_{\mathbb R^2}\varrho\mathrm d\varrho\mathrm d\vartheta\ W_{\varrho, \vartheta}'\mathscr Ff'(\varrho, \vartheta)\\
	&= (2\pi)^{-1}\int_{\mathbb R^2}\varrho\mathrm d\varrho\mathrm d\vartheta\ \sum_{j \in \mathbb Z}\sum_{n = 0}^\infty\ell_{|j|, n}(\epsilon^{1/2}\varrho)e^{ij\vartheta}\mathcal E_{j, n}\mathscr Ff'(\varrho, \vartheta)\\
	&= (2\pi)^{-1}\sum_{j \in \mathbb Z}\sum_{n = 0}^\infty\mathcal E_{j, n}\int_0^{2\pi}\mathrm d\vartheta\ e^{ij\vartheta}\int_0^\infty\varrho\mathrm d\varrho\ \ell_{|j|, n}(\epsilon^{1/2}\varrho)\mathscr Ff'(\varrho, \vartheta).
	\end{align}
	Therefore
	\begin{align}
	\mathscr W\phi_\lambda &= (2\pi)^{-1}\sum_{j \in \mathbb Z}\sum_{n = 0}^\infty\mathcal E_{j, n}\int_0^{2\pi}\mathrm d\vartheta\ e^{ij\vartheta}\int_0^\infty\varrho\mathrm d\varrho\ \ell_{|j|, n}(\epsilon^{1/2}\varrho)\mathscr F\phi'_\lambda(\varrho, \vartheta)\\
	&= (2\pi)^{-1}\sum_{j \in \mathbb Z}\sum_{n = 0}^\infty\mathcal E_{j, n}\int_0^{2\pi}\mathrm d\vartheta\ e^{ij\vartheta}\int_0^\infty\varrho\mathrm d\varrho\ \ell_{|j|, n}(\epsilon^{1/2}\varrho)\mathscr F^2\varphi'_\lambda(\varrho, \vartheta)\\
	&= (2\pi)^{-1}\sum_{j \in \mathbb Z}\sum_{n = 0}^\infty\mathcal E_{j, n}\int_0^{2\pi}\mathrm d\vartheta\ e^{ij\vartheta}\int_0^\infty\varrho\mathrm d\varrho\ \ell_{|j|, n}(\epsilon^{1/2}\varrho)\varphi'_\lambda(\varrho, \vartheta +\pi)
	\end{align}
	\begin{align}
	&= (2\pi)^{-1}\sum_{j \in \mathbb Z}\sum_{n = 0}^\infty\mathcal E_{j, n}\int_0^{2\pi}\mathrm d\vartheta\ e^{ij\vartheta}\int_0^\infty\varrho\mathrm d\varrho\ \ell_{|j|, n}(\epsilon^{1/2}\varrho)\\
	&\qquad \times\{\lambda^{-1/2}[\delta(\varrho -\lambda^{1/2}) +\delta(\varrho +\lambda^{1/2})]\varphi_\lambda'(\vartheta +\pi)\}\\
	&= (2\pi)^{-1}\sum_{j \in \mathbb Z}\sum_{n = 0}^\infty\mathcal E_{j, n}\ell_{|j|, n}(\epsilon^{1/2}\lambda^{1/2})\int_0^{2\pi}\mathrm d\vartheta\ e^{ij\vartheta}\varphi_\lambda'(\vartheta +\pi)\\
	&= (2\pi)^{-1}\sum_{j \in \mathbb Z}\sum_{n = 0}^\infty\mathcal E_{j, n}\ell_{|j|, n}(\epsilon^{1/2}\lambda^{1/2})\int_{-\pi}^\pi\mathrm d\vartheta\ e^{ij(\vartheta -\pi)}\varphi_\lambda'(\vartheta)
	\end{align}
	\begin{align}
	&= (2\pi)^{-1}\sum_{j \in \mathbb Z}\sum_{n = 0}^\infty\mathcal E_{j, n}\ell_{|j|, n}(\epsilon^{1/2}\lambda^{1/2})e^{-i\pi j}\int_{-\pi}^\pi\mathrm d\vartheta\ e^{ij\vartheta}\varphi_\lambda'(\vartheta)\\
	&= (2\pi)^{-1}\sum_{j \in \mathbb Z}\sum_{n = 0}^\infty\mathcal E_{j, n}\ell_{|j|, n}(\epsilon^{1/2}\lambda^{1/2})(-1)^j\int_0^{2\pi}\mathrm d\vartheta\ e^{-i(-j)\vartheta}\varphi_\lambda'(\vartheta)\\
	&= (2\pi)^{-1/2}\sum_{j \in \mathbb Z}(-1)^j\phi_{\lambda, -j}\sum_{n = 0}^\infty\ell_{|j|, n}(\epsilon^{1/2}\lambda^{1/2})\mathcal E_{j, n}\\
	&= (2\pi)^{-1/2}\sum_{j \in \mathbb Z}(-1)^j\phi_{\lambda, j}\sum_{n = 0}^\infty\ell_{|j|, n}(\epsilon^{1/2}\lambda^{1/2})\mathcal E_{-j, n}
	\end{align}
	Therefore
	\begin{align}
	\Phi_{\lambda, j} &= \phi_{\lambda, j}.
	\end{align}
	\end{proof}

	\section{The rotation operator}
	
	\begin{defn}
	We define:
	\begin{align}
	J &:= m_xp_y -m_yp_x,\quad \mathcal J := \mathscr WJ\mathscr W^{-1}.
	\end{align}
	\end{defn}
	
	\begin{prop}
	It is the case that
	\begin{align}
	J &= -i\epsilon\partial_\theta,\quad J\omega_{\lambda, j} = j\omega_{\lambda, j},\\
	\mathcal J &= 2^{-1}\mathrm{ad}_{R^2},\quad \mathcal J\Omega_{\lambda, j} = j\Omega_{\lambda, j},
	\end{align}
	for all $\lambda, j$, where $\zeta = x +iy = \rho e^{i\theta}$.
	\end{prop}
	
	\begin{proof}
	We note that
	\begin{align}
	\partial_\zeta &= 2^{-1}\partial_x -2^{-1}i\partial_y = 2^{-1}e^{-i\theta}\partial_\rho -i2^{-1}\rho^{-1}e^{-i\theta}\partial_\theta,\\
	\partial_{\overline \zeta} &= 2^{-1}\partial_x +2^{-1}i\partial_y = 2^{-1}e^{i\theta}\partial_\rho +i2^{-1}\rho^{-1}e^{i\theta}\partial_\theta.
	\end{align}
	
	\begin{align}
	2\zeta\partial_\zeta &= \rho\partial_\rho -i\partial_\theta,\\
	2\overline \zeta\partial_{\overline \zeta} &= \rho\partial_\rho +i\partial_\theta.
	\end{align}
	
	\begin{align}
	-2i\partial_\theta &= 2\zeta\partial_\zeta -2\overline \zeta\partial_{\overline \zeta}\\
	-i\partial_\theta &= \zeta\partial_\zeta -\overline \zeta\partial_{\overline \zeta}\\
	-i\partial_\theta &= 2^{-1}(x +iy)(\partial_x -i\partial_y) -2^{-1}(x -iy)(\partial_x +i\partial_y)\\
	&= -ix\partial_y +iy\partial_x\\
	-i\epsilon\partial_\theta f &= (m_xp_y -m_yp_x)f,\\
	-i\epsilon\partial_\theta &= m_xp_y -m_yp_x = J.
	\end{align}
	
	Since the Weyl transform can act as a unitary bijection, one has that
	\begin{align}
	\mathcal J A &= (M_XP_Y -M_YP_X)A\\
	&= (-M_X\mathrm{ad}_X -M_Y\mathrm{ad}_Y)A\\
	&= -[M_X(XA -AX) +M_Y(YA -AY)]\\
	&= -2^{-1}[X(XA -AX) +(XA -AX)X\\
	&\qquad +Y(YA -AY) +(YA -AY)Y]\\
	&= -2^{-1}[X^2A -XAX +XAX -AX^2\\
	&\qquad +Y^2A -YAY +YAY -AY^2]\\
	&= 2^{-1}[X^2A +Y^2A -AX^2 -AY^2]\\
	&= 2^{-1}[X^2 +Y^2, A]\\
	&= 2^{-1}\mathrm{ad}_{X^2 +Y^2}A.
	\end{align}
	Therefore
	\begin{align}
	\mathcal J &= 2^{-1}\mathrm{ad}_{R^2}.
	\end{align}
	
	By inspection it is clear that
	\begin{align}
	J\omega_{\lambda, j} &= j\omega_{\lambda, j},
	\end{align}
	for all $\lambda$, and therefore by unitary bijection property of the Weyl transform it must also be the case that
	\begin{align}
	\mathcal J\Omega_{\lambda, j} &= j\Omega_{\lambda, j},
	\end{align}
	for all $\lambda$.
	\end{proof}
	The above result is not original in principle due to the previous results of \cite{Ge84}\cite{Gr80} and review in \cite{Th98}. We consider, however, that the statement of this result explicitly in the form of an eigenfunction problem for the generator of rotations in the noncommutative plane to be novel. This result presents an alternative approach to the justification of $j = m -n$, considered from $E_{m, n}$, as an index of angular momentum studied in \cite{Ac13}.
	
	Furthermore the above result has interesting connections to the so-called \emph{metaplectic group}, see e.g. \cite{Fo89}, which is to the symplectic group what the spin group is to the special orthogonal group. In place of a full discussion hereof we give a very brief conceptual summary. Representations of the spin group are often called spinor representations of the special orthogonal group, elements of which are tensors built out of the so-called \emph{spinors}. These spinors are noteworthy for being ``square roots'' of spatial vectors, i.e. elements of the representation space of the special orthogonal group, in that a map sends the spatial vectors to elements of an operator algebra on the spinors and thereby the spatial vectors can be represented as sums of tensor products of spinors with dual spinors. This is the algebraic construction of spinors. In our setting there are analogous, albeit infinite dimensional, elements. If one considers $\mathscr C$ as not an algebra but a linear space alone, then the Weyl transform sends vectors, here elements of $\mathscr C$, to elements of an operator algebra, here $\mathscr N$, that act on an auxiliary linear space, here $\mathscr H$. In this sense elements of $\mathscr C$ are \emph{symplectic vectors} and elements of $\mathscr H$ are \emph{symplectic spinors}.
	
	There also exists a representation theory construction of spinors, whereby one considers the spinors as ``hidden'' representations of the special orthogonal group. This is often encountered in the theory of angular momentum. There the angular momentum operator is essentially the rotation operator. Integer weight vectors of this operator are the spatial vectors and the half-integer weight vectors thereof are the spinors. We have seen that $J$ is a rotation operator and the $j = m -n$ are its eigenvalues. One could then interpret the $j$ as integer quantized units of \emph{symplectic angular momentum}. One might further expect that the $h_n$ should be realizable as half-integer weight vectors, i.e. symplectic spinors of half-integer symplectic angular momentum. The concrete realization of this assertion requires the machinery of representation theory of nilpotent Lie algebras and is therefore beyond the scope of this paper. Instead we present two highly suggestive observations.
	
	First, since
	\begin{align}
	X^2 +Y^2 = \epsilon(2\mathcal S^\dag\mathcal S +I),
	\end{align}
	it is the case that
	\begin{align}
	2^{-1}R^2 = 2^{-1}(X^2 +Y^2) = \epsilon(\mathcal S^\dag\mathcal S +2^{-1}I)
	\end{align}
	and therefore
	\begin{align}
	2^{-1}R^2h_n = \epsilon(n +2^{-1})h_n.
	\end{align}
	Second, we recall that the Laguerre functions appear in our calculations as $l^{(|j|)}_n(\epsilon\rho^2/2)$, so the additional parameter encodes absolute angular momentum. We then simply observe a classical orthogonal polynomial relation restated in terms of normalized orthogonal functions.
	\begin{prop}
	It is the case that
	\begin{align}
	\epsilon^{-1/4}h_{2n}(\epsilon^{1/2}x) &= (-1)^nn!l^{(-1/2)}_n(x^2),\quad \epsilon^{-1/4}h_{2n +1}(\epsilon^{1/2}x) = x(-1)^nn!l^{(1/2)}_n(x^2),
	\end{align}
	for all $n, x$.
	\end{prop}
	
	\begin{proof}
	\cite[p.~256]{AS}:
	\begin{align}
	\Gamma(2\zeta) &= (2\pi)^{-1/2}2^{2\zeta -1/2}\Gamma(\zeta)\Gamma(\zeta +1/2)\\
	\Gamma(\zeta +1/2) &= (2\pi)^{1/2}2^{-2\zeta +1/2}\Gamma(2\zeta)/\Gamma(\zeta)\\
	&= 2\pi^{1/2}2^{-2\zeta}\Gamma(2\zeta)/\Gamma(\zeta).
	\end{align}
	
	\begin{align}
	L_n^{(\alpha)}(x) &= \sum_{j = 0}^n(-1)^j\binom{n +\alpha}{n -j}\frac{x^j}{j!} = \sum_{j = 0}^n(-1)^j\frac{(n +\alpha)!}{(n -j)!(\alpha +j)!j!}x^j.
	\end{align}
	
	\cite[\S~18.5]{DLMF}:
	\begin{align}
	L^{(\alpha)}_n(x) &= \sum_{j = 0}^n[(n -j)!j!]^{-1}(\alpha +j +1)_{n -j}(-x)^j.
	\end{align}
	\begin{align}
	H_n(x) &= n!\sum_{j = 0}^{\lfloor n/2\rfloor}(-1)^j[j!(n -2j)!]^{-1}(2x)^{n -2j}.
	\end{align}
	
	For $n = 2m$:
	\begin{align}
	H_{2m}(x) &= (2m)!\sum_{j = 0}^{\lfloor m\rfloor}(-1)^j[j!(2m -2j)!]^{-1}(2x)^{2m -2j},\\
	&= (2m)!\sum_{j = 0}^m(-1)^j[j!(2m -2j)!]^{-1}(2x)^{2m -2j}.
	\end{align}
	
	For $n = 2m +1$:
	\begin{align}
	H_{2m +1}(x) &= (2m +1)!\sum_{j = 0}^{\lfloor m +1/2\rfloor}(-1)^j[j!(2m +1 -2j)!]^{-1}(2x)^{2m +1 -2j}\\
	&= 2x(2m +1)!\sum_{j = 0}^m(-1)^j[j!(2m +1 -2j)!]^{-1}(2x)^{2m -2j}\\
	\end{align}
	
	\cite[\S~5.2]{DLMF}:
	\begin{align}
	(a)_n &= a(a +1)(a +2)\cdots(a +n -1)\\
	(a)_n &= \Gamma(a +n)/\Gamma(a),\quad 0 \le a \in \mathbb Z,\\
	(a +1)_n &= \Gamma(a +n +1)/\Gamma(a +1) = (a +n)!/a!,\quad 0 \le a \in \mathbb Z.
	\end{align}
	
	\begin{align}
	(\alpha +j +1)_{n -j} &= \Gamma(\alpha +j +n -j +1)/\Gamma(\alpha +j +1)\\
	&= \Gamma(\alpha +n +1)/\Gamma(\alpha +j +1)\\
	&= [\Gamma(j +1 +\alpha)]^{-1}\Gamma(n +1 +\alpha).
	\end{align}
	
	\begin{align}
	(j +1 +1/2)_{n -j} &= [\Gamma(j +1 +1/2)]^{-1}\Gamma(n +1 +1/2)\\
	&= [2\pi^{1/2}2^{-2(j +1)}\Gamma(2\{j +1\})/\Gamma(j +1)]^{-1}\\
	&\qquad \times 2\pi^{1/2}2^{-2\{n +1\}}\Gamma(2\{n +1\})/\Gamma(n +1)\\
	&= 2^{-2(n -j)}[\Gamma(2j +2)\Gamma(n +1)]^{-1}\Gamma(j +1)\Gamma(2n +2)\\
	&= 2^{-2(n -j)}[(2j +1)!n!]^{-1}j!(2n +1)!
	\end{align}
	
	\begin{align}
	(j +1/2)_{n -j} &= [\Gamma(j +1/2)]^{-1}\Gamma(n +1/2)\\
	&= [2\pi^{1/2}2^{-2j}\Gamma(2j)/\Gamma(j)]^{-1}2\pi^{1/2}2^{-2n}\Gamma(2n)/\Gamma(n)\\
	&= 2^{-2(n -j)}[\Gamma(2j)\Gamma(n)]^{-1}\Gamma(j)\Gamma(2n)\\
	&= 2^{-2(n -j)}[\Gamma(2j +1)\Gamma(n +1)/2jn]^{-1}\Gamma(j +1)\Gamma(2n +1)/2jn\\
	&= 2^{-2(n -j)}[\Gamma(2j +1)\Gamma(n +1)]^{-1}\Gamma(j +1)\Gamma(2n +1)\\
	&= 2^{-2(n -j)}[(2j)!n!]^{-1}j!(2n)!.
	\end{align}
	
	\begin{align}
	L_n^{(1/2)}(x) &= \sum_{j = 0}^n[(n -j)!j!]^{-1}(j +1 +1/2)_{n -j}(-x)^j\\
	&= \sum_{j = 0}^n[(n -j)!j!]^{-1}2^{-2(n -j)}[(2j +1)!n!]^{-1}j!(2n +1)!(-x)^j\\
	&= 2^{-2n}(n!)^{-1}(2n +1)!\sum_{j = 0}^n2^{2j}[(2j +1)!(n -j)!]^{-1}(-x)^j
	\end{align}
	
	\begin{align}
	&= 2^{-2n}(n!)^{-1}(2n +1)!\sum_{j = 0}^n[(2j +1)!(n -j)!]^{-1}(-2^2x)^j\\
	&= 2^{-2n}(n!)^{-1}(2n +1)!\sum_{j = 0}^n(-1)^j[(2j +1)!(n -j)!]^{-1}(2x^{1/2})^{2j}\\
	&= 2^{-2n}(n!)^{-1}(2n +1)!\sum_{k = 0}^n(-1)^{n -k}[(2n -2k +1)!k!]^{-1}(2x^{1/2})^{2n -2k},\\
	&\qquad j = n -k,
	\end{align}
	
	\begin{align}
	&= 2^{-2n}(n!)^{-1}(2n +1)!\sum_{k = 0}^n(-1)^{n -k}[k!(2n +1 -2k)!]^{-1}(2x^{1/2})^{2n -2k}\\
	&= (-2^{-2})^n(n!)^{-1}(2n +1)!\sum_{k = 0}^n(-1)^k[k!(2n +1 -2k)!]^{-1}(2x^{1/2})^{2n -2k}
	\end{align}
	
	\begin{align}
	L_m^{(1/2)}(x) &= (-2^{-2})^m(m!)^{-1}(2m +1)!\sum_{j = 0}^m(-1)^j[j!(2m +1 -2j)!]^{-1}(2x^{1/2})^{2m -2j}.
	\end{align}
	
	\begin{align}
	L^{(-1/2)}_n(x) &= \sum_{j = 0}^n[(n -j)!j!]^{-1}(j +1/2)_{n -j}(-x)^j\\
	&= \sum_{j = 0}^n[(n -j)!j!]^{-1}2^{-2(n -j)}[(2j)!n!]^{-1}j!(2n)!(-x)^j\\
	&= (n!)^{-1}(2n)!\sum_{j = 0}^n[(2j)!(n -j)!]^{-1}2^{-2(n -j)}(-x)^j
	\end{align}
	
	\begin{align}
	&= 2^{-2n}(n!)^{-1}(2n)!\sum_{j = 0}^n(-1)^j[(2j)!(n -j)!]^{-1}2^{2j}(x)^j\\
	&= 2^{-2n}(n!)^{-1}(2n)!\sum_{j = 0}^n(-1)^j[(2j)!(n -j)!]^{-1}(2x^{1/2})^{2j}\\
	&= 2^{-2n}(n!)^{-1}(2n)!\sum_{k = 0}^n(-1)^{n -k}[(2n -2k)!k!]^{-1}(2x^{1/2})^{2n -2k}\\
	&\qquad j = n -k,
	\end{align}
	
	\begin{align}
	&= 2^{-2n}(n!)^{-1}(2n)!\sum_{k = 0}^n(-1)^{n -k}[k!(2n -2k)!]^{-1}(2x^{1/2})^{2n -2k}\\
	&= (-2^{-2})^n(n!)^{-1}(2n)!\sum_{k = 0}^n(-1)^k[k!(2n -2k)!]^{-1}(2x^{1/2})^{2n -2k}
	\end{align}
	
	\begin{align}
	L^{(-1/2)}_m(x) &= (-2^{-2})^m(m!)^{-1}(2m)!\sum_{j = 0}^m(-1)^j[j!(2m -2j)!]^{-1}(2x^{1/2})^{2m -2j}.
	\end{align}
	
	We observe that from the above one has:
	\begin{align}
	H_{2m}(x) &= (2m)!\sum_{j = 0}^m(-1)^j[j!(2m -2j)!]^{-1}(2x)^{2m -2j},\\
	H_{2m +1}(x) &= 2x(2m +1)!\sum_{j = 0}^m(-1)^j[j!(2m +1 -2j)!]^{-1}(2x)^{2m -2j},\\
	L_m^{(1/2)}(x) &= (-2^{-2})^m(m!)^{-1}(2m +1)!\sum_{j = 0}^m(-1)^j[j!(2m +1 -2j)!]^{-1}(2x^{1/2})^{2m -2j},\\
	L^{(-1/2)}_m(x) &= (-2^{-2})^m(m!)^{-1}(2m)!\sum_{j = 0}^m(-1)^j[j!(2m -2j)!]^{-1}(2x^{1/2})^{2m -2j},
	\end{align}
	and therefore
	\begin{align}
	L^{(-1/2)}_n(x^2) &= (-2^{-2})^n(n!)^{-1}H_{2n}(x),\\
	L_n^{(1/2)}(x^2) &= (-2^{-2})^n(n!)^{-1}(2x)^{-1}H_{2n +1}(x).
	\end{align}
	\begin{align}
	H_{2n}(x) &= (-2^2)^nn!L^{(-1/2)}_n(x^2),\\
	H_{2n +1}(x) &= (-2^2)^nn!(2x)L_n^{(1/2)}(x^2).
	\end{align}
	
	\begin{align}
	&l^{(\alpha)}_n(x) := \Gamma(n +\alpha +1)^{-1/2}\Gamma(n +1)^{1/2}e^{-x/2}L^{(\alpha)}_n(x).
	\end{align}
	
	\begin{align}
	l^{(1/2)}_n(x) &= \Gamma(n +1 +1/2)^{-1/2}\Gamma(n +1)^{1/2}e^{-x/2}L^{(1/2)}_n(x)\\
	&= [2\pi^{1/2}2^{-2(n +1)}\Gamma(2n +2)/\Gamma(n +1)]^{-1/2}\Gamma(n +1)^{1/2}e^{-x/2}L^{(1/2)}_n(x)\\
	&= \pi^{-1/4}2^{n +1/2}[(2n +1)!]^{-1/2}e^{-x/2}L^{(1/2)}_n(x).
	\end{align}
	
	\begin{align}
	l^{(-1/2)}_n(x) &= \Gamma(n +1/2)^{-1/2}\Gamma(n +1)^{1/2}e^{-x/2}L^{(-1/2)}_n(x)\\
	&= [2\pi^{1/2}2^{-2n}\Gamma(2n +1)n/2n\Gamma(n +1)]^{-1/2}[\Gamma(n +1)]^{1/2}e^{-x/2}L^{(-1/2)}_n(x)\\
	&= \pi^{-1/4}2^n[(2n)!]^{-1/2}e^{-x/2}L^{(-1/2)}_n(x).
	\end{align}
	
	\begin{align}
	L^{(1/2)}_n(x) &= \pi^{1/4}2^{-(n +1/2)}[(2n +1)!]^{1/2}e^{x/2}l^{(1/2)}_n(x)\\
	L^{(-1/2)}_n(x) &= \pi^{1/4}2^{-n}[(2n)!]^{1/2}e^{x/2}l^{(-1/2)}_n(x).
	\end{align}
	
	\begin{align}
	\epsilon^{-1/4}h_n(\epsilon^{1/2}x) &= \pi^{-1/4}2^{-n/2}(n!)^{-1/2}e^{-x^2/2}H_n(x).
	\end{align}
	
	\begin{align}
	\epsilon^{-1/4}h_{2m}(\epsilon^{1/2}x) &= \pi^{-1/4}2^{-m}[(2m)!]^{-1/2}e^{-x^2/2}H_{2m}(x),\\
	\epsilon^{-1/4}h_{2m +1}(\epsilon^{1/2}x) &= \pi^{-1/4}2^{-m -1/2}[(2m +1)!]^{-1/2}e^{-x^2/2}H_{2m +1}(x).
	\end{align}
	
	\begin{align}
	\epsilon^{-1/4}h_{2n}(\epsilon^{1/2}x) &= \pi^{-1/4}2^{-n}[(2n)!]^{-1/2}e^{-x^2/2}H_{2n}(x),\\
	\epsilon^{-1/4}h_{2n +1}(\epsilon^{1/2}x) &= \pi^{-1/4}2^{-n -1/2}[(2n +1)!]^{-1/2}e^{-x^2/2}H_{2n +1}(x).
	\end{align}
	
	\begin{align}
	\epsilon^{-1/4}h_{2n}(\epsilon^{1/2}x) &= \pi^{-1/4}2^{-n}[(2n)!]^{-1/2}e^{-x^2/2}H_{2n}(x)\\
	&= \pi^{-1/4}2^{-n}[(2n)!]^{-1/2}e^{-x^2/2}(-2^2)^nn!L^{(-1/2)}_n(x^2)\\
	&= \pi^{-1/4}(-2)^n[(2n)!]^{-1/2}n!e^{-x^2/2}L^{(-1/2)}_n(x^2)
	\end{align}
	
	\begin{align}
	&= \pi^{-1/4}(-2)^n[(2n)!]^{-1/2}n!e^{-x^2/2}\\
	&\qquad \times\pi^{1/4}2^{-n}[(2n)!]^{1/2}e^{x^2/2}l^{(-1/2)}_n(x^2)\\
	&= (-1)^nn!l^{(-1/2)}_n(x^2).
	\end{align}
	
	\begin{align}
	\epsilon^{-1/4}h_{2n +1}(\epsilon^{1/2}x) &= \pi^{-1/4}2^{-n -1/2}[(2n +1)!]^{-1/2}e^{-x^2/2}H_{2n +1}(x)\\
	&= \pi^{-1/4}2^{-n -1/2}[(2n +1)!]^{-1/2}e^{-x^2/2}(-2^2)^nn!(2x)L_n^{(1/2)}(x^2)\\
	&= (-1)^n\pi^{-1/4}2^{n -1/2}[(2n +1)!]^{-1/2}n!e^{-x^2/2}(2x)L_n^{(1/2)}(x^2)
	\end{align}
	
	\begin{align}
	&= (-1)^n\pi^{-1/4}2^{n -1/2}[(2n +1)!]^{-1/2}n!e^{-x^2/2}(2x)\\
	&\qquad \times\pi^{1/4}2^{-(n +1/2)}[(2n +1)!]^{1/2}e^{x^2/2}l^{(1/2)}_n(x^2)\\
	&= (-1)^n2^{-1}n!(2x)l^{(1/2)}_n(x^2)\\
	&= x(-1)^nn!l^{(1/2)}_n(x^2).
	\end{align}
	\end{proof}

	\section{Elementary polar functions}
	
	\begin{prop}
	\begin{align}
	\mathscr W(\rho^{|j|}) = R^{|j|} = \epsilon^{|j|/2}\sum_{n = 0}^\infty(2n +1)^{|j|/2}\mathcal E_{0, n}.
	\end{align}
	\end{prop}
	
	\begin{proof}
	\begin{align}
	&R^2h_n = \epsilon(2n +1)h_n.
	\end{align}
	Then by the spectral theorem:
	\begin{align}
	&R^{|j|}\mathcal E_{j, n} = \epsilon^{|j|/2}\sum_{n = 0}^\infty(2n +1)^{|j|/2}\mathcal E_{0, n},
	\end{align}
	and therefore the correspondence $\mathscr W(\rho^{|j|}) = R^{|j|}$ is well posed and the above formula gives the image of the transform explicitly.
	\end{proof}
	
	\begin{defn}
	We define
	\begin{align}
	\mathcal T_j &:= 2^{|j|/2}\sum_{n = 0}^\infty(n!)^{-1/2}[(n +|j|)!]^{1/2}(2n +|j| +1)^{-|j|/2}\mathcal E_{j, n}.
	\end{align}
	\end{defn}
	
	\begin{prop}
	It is the case that
	\begin{align}
	\mathscr W(e^{-ij\theta}) = \mathcal T_j.
	\end{align}
	\end{prop}
	
	\begin{proof}
	We recall that:
	\begin{align}
	&R^2h_n = \epsilon(2n +1)h_n,\\
	&M_{R^2}E_{m, n} = 2^{-1}(R^2E_{m, n} +E_{m, n}R^2) = \epsilon(m +n +1)E_{m, n},\\
	&M_{R^2}\mathcal E_{j, n} = \epsilon(2n +|j| +1)\mathcal E_{j, n},
	\end{align}
	Then by the spectral theorem:
	\begin{align}
	&M_{R^2}^{-|j|/2}\mathcal E_{j, n} = \epsilon^{-|j|/2}(2n +|j| +1)^{-|j|/2}\mathcal E_{j, n}.
	\end{align}
	
	We recall that
	\begin{align}
	&\omega_{\lambda, j}'(\rho, \theta) := i^{-|j|}J_{|j|}(\lambda^{1/2}\rho)e^{ij\theta},\\
	&\Omega_{\lambda, j} := (-1)^j\sum_{n = 0}^\infty\ell_{|j|, n}(\epsilon^{1/2}\lambda^{1/2})\mathcal E_{-j, n},
	\end{align}
	\begin{align}
	\omega_{\lambda, j}'(\rho, \theta) &= \mathscr W^{-1}\Omega_{\lambda, j}'(\rho, \theta)\\
	i^{-|j|}J_{|j|}(\lambda^{1/2}\rho)e^{ij\theta} &= \mathscr W^{-1}[(-1)^j\sum_{n = 0}^\infty\ell_{|j|, n}(\epsilon^{1/2}\lambda^{1/2})\mathcal E_{-j, n}]'(\rho, \theta)\\
	&= (-1)^j\sum_{n = 0}^\infty\ell_{|j|, n}(\epsilon^{1/2}\lambda^{1/2})\mathscr W^{-1}\mathcal E_{-j, n}'(\rho, \theta)
	\end{align}
	
	We note that \cite[\S~10.2]{DLMF}:
	\begin{align}
	J_j(x) = (x/2)^j\sum_{k = 0}^\infty\frac{(-1)^k(x/2)^{2k} }{k!\Gamma(k +j +1)},
	\end{align}
	therefore
	\begin{align}
	i^{-|j|}J_{|j|}(\lambda^{1/2}\rho) = i^{-|j|}(\lambda^{1/2}\rho/2)^{|j|}\sum_{k = 0}^\infty\frac{(-1)^k(\lambda^{1/2}\rho/2)^{2k} }{k!(k +|j|)!},
	\end{align}
	and
	\begin{align}
	\lim_{\lambda \searrow 0}[\lambda^{-|j|/2}i^{-|j|}J_{|j|}(\lambda^{1/2}\rho)] = [(2i)^{|j|}(|j|)!]^{-1}\rho^{|j|},\\
	[(2i)^{|j|}(|j|)!]\lim_{\lambda \searrow 0}[\lambda^{-|j|/2}i^{-|j|}J_{|j|}(\lambda^{1/2}\rho)] = \rho^{|j|}.
	\end{align}
	
	We recall that
	\begin{align}
	\ell_{j, n}(\rho) &= \ell_{|j|, n}(\rho) = (2^{-1/2}i\rho)^{|j|}l^{(|j|)}_n(\rho^2/2),
	\end{align}
	\begin{align}
	l^{(\alpha)}_n(x) &= \Gamma(n +\alpha +1)^{-1/2}\Gamma(n +1)^{1/2}e^{-x/2}L^{(\alpha)}_n(x),\\
	L^{(\alpha)}_n(x) &= \sum_{j = 0}^n\binom{n +\alpha}{j +\alpha}\frac{(-x)^j}{j!},\\
	L^{(\alpha)}_n(0) &= \binom{n +\alpha}n = \Gamma(n +1)^{-1}\Gamma(\alpha +1)^{-1}\Gamma(n +\alpha +1)
	\end{align}
	\begin{align}
	l^{(|j|)}_n(x) &= [(n +|j|)!]^{-1/2}(n!)^{1/2}e^{-x/2}L^{(|j|)}_n(x),\\
	L^{(|j|)}_n(x) &= \sum_{k = 0}^n\binom{n +|j|}{k +|j|}\frac{(-x)^k}{k!},\\
	L^{(|j|)}_n(0) &= \binom{n +|j|}n = (n!)^{-1}(|j|!)^{-1}[(n +|j|)!],
	\end{align}
	Therefore
	\begin{align}
	l^{(|j|)}_n(0) &= [(n +|j|)!]^{-1/2}(n!)^{1/2}L^{(|j|)}_n(0)\\
	&= [(n +|j|)!]^{-1/2}(n!)^{1/2}(n!)^{-1}(|j|!)^{-1}[(n +|j|)!]\\
	&= (n!)^{-1/2}(|j|!)^{-1}[(n +|j|)!]^{1/2},
	\end{align}
	and
	\begin{align}
	\lim_{\lambda \searrow 0}[\lambda^{-|j|/2}\ell_{|j|, n}(\epsilon^{1/2}\lambda^{1/2})] &= (2^{-1/2}i\epsilon^{1/2})^{|j|}l^{(|j|)}_n(0)\\
	&= (2^{-1/2}i\epsilon^{1/2})^{|j|}(n!)^{-1/2}(|j|!)^{-1}[(n +|j|)!]^{1/2},
	\end{align}
	\begin{align}
	&[(2i)^{|j|}(|j|)!]\lim_{\lambda \searrow 0}[\lambda^{-|j|/2}\ell_{|j|, n}(\epsilon^{1/2}\lambda^{1/2})]\\
	&\qquad = [(2i)^{|j|}(|j|)!](2^{-1/2}i\epsilon^{1/2})^{|j|}(n!)^{-1/2}(|j|!)^{-1}[(n +|j|)!]^{1/2}\\
	&\qquad = (-1)^j(2\epsilon)^{|j|/2}(n!)^{-1/2}[(n +|j|)!]^{1/2}.
	\end{align}
	
	Then
	\begin{align}
	\omega_{\lambda, j}'(\rho, \theta) &= \mathscr W^{-1}\Omega_{\lambda, j}'(\rho, \theta)\\
	i^{-|j|}J_{|j|}(\lambda^{1/2}\rho)e^{ij\theta} &= (-1)^j\sum_{n = 0}^\infty\ell_{|j|, n}(\epsilon^{1/2}\lambda^{1/2})\mathscr W^{-1}\mathcal E_{-j, n}'(\rho, \theta)
	\end{align}
	\begin{align}
	&[(2i)^{|j|}(|j|)!]\lim_{\lambda \searrow 0}[\lambda^{-|j|/2}i^{-|j|}J_{|j|}(\lambda^{1/2}\rho)e^{ij\theta}]\\
	&\qquad = [(2i)^{|j|}(|j|)!]\lim_{\lambda \searrow 0}[\lambda^{-|j|/2}(-1)^j\sum_{n = 0}^\infty\ell_{|j|, n}(\epsilon^{1/2}\lambda^{1/2})\mathscr W^{-1}\mathcal E_{-j, n}'(\rho, \theta)]
	\end{align}
	
	\begin{align}
	&[(2i)^{|j|}(|j|)!]\lim_{\lambda \searrow 0}[\lambda^{-|j|/2}i^{-|j|}J_{|j|}(\lambda^{1/2}\rho)]e^{ij\theta}\\
	&\qquad = (-1)^j\sum_{n = 0}^\infty[(2i)^{|j|}(|j|)!]\lim_{\lambda \searrow 0}[\lambda^{-|j|/2}\ell_{|j|, n}(\epsilon^{1/2}\lambda^{1/2})]\mathscr W^{-1}\mathcal E_{-j, n}'(\rho, \theta)
	\end{align}
	
	\begin{align}
	\rho^{|j|}e^{ij\theta} &= (-1)^j\sum_{n = 0}^\infty(-1)^j(2\epsilon)^{|j|/2}(n!)^{-1/2}[(n +|j|)!]^{1/2}\mathscr W^{-1}\mathcal E_{-j, n}'(\rho, \theta)\\
	&= (2\epsilon)^{|j|/2}\sum_{n = 0}^\infty(n!)^{-1/2}[(n +|j|)!]^{1/2}\mathscr W^{-1}\mathcal E_{-j, n}'(\rho, \theta)
	\end{align}
	
	\begin{align}
	e^{ij\theta} &= (2\epsilon)^{|j|/2}\sum_{n = 0}^\infty(n!)^{-1/2}[(n +|j|)!]^{1/2}\mathscr W^{-1}(M_{R^2}^{-|j|/2}\mathcal E_{-j, n})'(\rho, \theta)\\
	&= (2\epsilon)^{|j|/2}\sum_{n = 0}^\infty(n!)^{-1/2}[(n +|j|)!]^{1/2}\mathscr W^{-1}[\epsilon^{-|j|/2}(2n +|j| +1)^{-|j|/2}\mathcal E_{-j, n}]'(\rho, \theta)\\
	&= 2^{|j|/2}\sum_{n = 0}^\infty(n!)^{-1/2}[(n +|j|)!]^{1/2}(2n +|j| +1)^{-|j|/2}\mathscr W^{-1}\mathcal E_{-j, n}'(\rho, \theta)
	\end{align}
	We then conclude that
	\begin{align}
	\mathscr W(e^{-ij\theta}) &= 2^{|j|/2}\sum_{n = 0}^\infty(n!)^{-1/2}[(n +|j|)!]^{1/2}(2n +|j| +1)^{-|j|/2}\mathcal E_{j, n} = \mathcal T_j.
	\end{align}
	\end{proof}
	
	\begin{rem}
	If one takes
	\begin{align}
	Z := (2\epsilon)^{1/2}\mathcal S,\quad Z^\dag := (2\epsilon)^{1/2}\mathcal S^\dag,
	\end{align}
	one may observe that
	\begin{align}
	Z^\dag = (2\epsilon)^{1/2}\sum_{n = 0}^\infty(n +1)^{1/2}E_{n +1, n},\quad Z = (2\epsilon)^{1/2}\sum_{n = 0}^\infty (n +1)^{1/2}E_{n, n +1},
	\end{align}
	and find
	\begin{align}
	e^{i\Theta} &= M_{R^2}^{-1/2}Z = (2\epsilon)^{1/2}\sum_{n = 0}^\infty(n +1)^{1/2}M_{R^2}^{-1/2}E_{n, n +1}\\
	&= (2\epsilon)^{1/2}\sum_{n = 0}^\infty(n +1)^{1/2}\epsilon^{-1/2}(2n +2)^{-1/2}E_{n, n +1}\\
	&= \sum_{n = 0}^\infty E_{n, n +1} = \mathcal T_{-1}
	\end{align}
	\begin{align}
	e^{-i\Theta} &= M_{R^2}^{-1/2}Z^\dag = (2\epsilon)^{1/2}\sum_{n = 0}^\infty(n +1)^{1/2}M_{R^2}^{-1/2}E_{n +1, n}\\
	&= (2\epsilon)^{1/2}\sum_{n = 0}^\infty(n +1)^{1/2}\epsilon^{-1/2}(2n +2)^{-1/2}E_{n +1, n}\\
	&= \sum_{n = 0}^\infty E_{n +1, n} = \mathcal T_1
	\end{align}
	\end{rem}

	\section{Connections between local decay estimates}
		
	Kostenko and Teschl, Theorem 6.3 of \cite{KoTe1602}, found that the estimate
	\begin{align}
	||e^{-i\mathcal D^{(\alpha)}t}||_{\ell^1(\mathbb Z_+, \sigma^{(\alpha)}) \to \ell^\infty(\mathbb Z_+, \sigma^{(\alpha), -1})} = (1 +t^2)^{-(1 +\alpha)/2},\quad \sigma^{(\alpha)}_n := [L^{(\alpha)}_n(0)]^{1/2},
	\end{align}
	holds for all $\alpha \in \mathbb R_+$. This estimate is an extension of their earlier estimate for the case $\alpha = 0$ found in \cite{KoTe1604}. They note that this appears to be an analogue of the continuum result found in \cite{KoTeTo15}\cite{KoTr14}:
	\begin{align}
	||e^{-iH^{(\alpha)}t}||_{L^1(\mathbb Z_+, x^j) \to \ell^\infty(\mathbb Z_+, x^{-\alpha})} = \mathcal O(|t|^{-\alpha -1/2}),\quad t \to \infty,
	\end{align}
	for all $-1/2 \le \alpha \in \mathbb Z$, where
	\begin{align}
	H^{(\alpha)} := -\partial_x^2 +\alpha(\alpha +1)x^{-2},
	\end{align}
	which bears a similarity to
	\begin{align}
	\lambda f_j &= lf_j = (-\partial_\rho^2 -\rho^{-1}\partial_\rho +j^2\rho^{-2})f_j,
	\end{align}
	where $f_j'(\rho, \theta) = g_j'(\rho)e^{ij\theta}$. We now consider a way in which such estimates for respectively discrete and continuous systems can be considered not just analogous but directly related.
	
	\begin{align}
	\mathcal G = \mathscr WG\mathscr W^{-1}
	\end{align}
	
	\begin{align}
	||\mathcal G||_{\ell^1(\mathbb Z_+, \sigma_j) \to \ell^\infty(\mathbb Z_+, \sigma_j^{-1})} = \sup_{m, n \in \mathbb Z_+}\sigma_{j, m}^{-1}\sigma_{j, n}^{-1}|\mathcal G_{m, n}|.
	\end{align}
	
	\begin{align}
	&\mathcal U = \sum_{j \in \mathbb Z}\sum_{n = 0}^\infty u_{j, n}\mathcal E_{j, n}\otimes\mathcal E^*_{j, n},\quad 0 < u_{j, n} \in \mathbb R, \forall j, n,\quad U = \mathscr W^{-1}\mathcal U\mathscr W.
	\end{align}
	
	\begin{align}
	U &= \mathscr W^{-1}\mathcal U\mathscr W = \mathscr W^{-1}(\sum_{n = 0}^\infty u{j, n}\mathcal E_{j, n}\otimes\mathcal E^*_{j, n})\mathscr W = \sum_{n = 0}^\infty u_{j, n}(\mathcal W^{-1}\mathcal E_{j, n})\otimes(\mathcal W^\dag\mathcal E_{j, n})^*\\
	&= \sum_{n = 0}^\infty u_{j, n}(\mathscr W^{-1}\mathcal E_{j, n})\otimes(\mathscr W^{-1}\mathcal E_{j, n})^* = \sum_{n = 0}^\infty u_{j, n}\varepsilon_{j, n}\otimes\varepsilon_{j, n}^*,
	\end{align}
	where we implemented $\mathscr W^{-1} = \mathscr W^\dag$ since the Weyl transform is unitary between suitably regular elements.
	
	\begin{align}
	\mathcal U_j &= \sum_{n = 0}^\infty u_{j, n}\mathcal E_{j, n}\otimes\mathcal E^*_{j, n},\quad U_j = \sum_{n = 0}^\infty u_{j, n}\varepsilon_{j, n}\otimes\varepsilon^*_{j, n},\\
	\mathcal U^{-1}_j &= \sum_{n = 0}^\infty u^{-1}_{j, n}\mathcal E_{j, n}\otimes\mathcal E^*_{j, n},\quad U^{-1}_j = \sum_{n = 0}^\infty u^{-1}_{j, n}\varepsilon_{j, n}\otimes\varepsilon^*_{j, n}.
	\end{align}
	
	\begin{align}
	&||\mathcal G||_{\ell^1(\mathbb Z_+, u_j) \to \ell^\infty(\mathbb Z_+, u_j^{-1})} = \sup_{m, n \in \mathbb Z_+}u_{j, m}^{-1}u_{j, n}^{-1}|\mathcal G_{m, n}| = \sup_{m, n \in \mathbb Z_+}|u_{j, m}^{-1}\mathcal G_{m, n}u_{j, n}^{-1}|\\
	&\qquad = \sup_{m, n \in \mathbb Z_+}|\langle\mathcal E_{j, m}, \mathcal U_j^{-1}\mathcal G\mathcal U_j^{-1}\mathcal E_{j, n}\rangle| = \sup_{m, n \in \mathbb Z_+}|\langle\varepsilon_{j, m}, U_j^{-1}GU_j^{-1}\varepsilon_{j, n}\rangle|\\
	&\qquad = \sup_{m, n \in \mathbb Z_+}|u_{j, m}^{-1}G_{m, n}u_{j, n}^{-1}| = \sup_{m, n \in \mathbb Z_+}u_{j, m}^{-1}u_{j, n}^{-1}|G_{m, n}|.
	\end{align}
	We conjecture that these correspondences may be tightened through implementation of previous results on $L^p$ estimates on the Heisenberg group, see e.g. \cite{Th93}\cite{St91}, and possibly with observations of correspondences between $\mathscr C$ and $\mathscr N$ in the spatially asymptotic sense.

	\ \\
	\thanks{We thank Tobias Hartnick and Amos Nevo for helpful advice and insights.}

	\end{document}